\documentclass[a4paper,11pt,oneside,openright]{article}
\usepackage{amsmath} % Redefines \[ and \] as \begin{equation*} and \end{equation*}, so they are more robust.
\usepackage{amssymb} % For mathfont blackboard bold (mathbb).
\usepackage{amscd} % Gives the CD (commutative diagram) environment.
\usepackage{amsthm} % Gives the proof environment.
\usepackage{graphicx}
\usepackage{rotating} % For the sidewaystable environment.
\usepackage{hyperref} % To have links on the pdf file.

\usepackage[a4paper,top=3cm,bottom=2.5cm,left=2cm,right=2cm]{geometry}

\newcommand{\bbR}{\mathbb{R}}
\newcommand{\bbE}{\mathbb{E}}
\newcommand{\bbS}{\mathbb{S}}

\newcommand{\calN}{\mathcal{N}}
\newcommand{\bfb}{\mathbf{b}}

\newcommand{\bfx}{\mathbf{x}}

\newcommand{\dd}{\mathrm{d}}

\newtheorem{thrm}{Theorem}[section]
\newtheorem{dfntn}{Definition}[section]
\newtheorem{crllr}[thrm]{Corollary}
\newtheorem{prpstn}[thrm]{Proposition}
\newtheorem{prop}[thrm]{Proposition}
\title{Entropy and efficiency of the ETF market}
\author{Lucio Maria Calcagnile, Fulvio Corsi, Stefano Marmi}
\date{September 2016}

\begin{document}
\maketitle
\abstract{We investigate the relative information efficiency of financial markets by measuring the entropy of the time series of high frequency data. Our tool to measure efficiency is the Shannon entropy, applied to 2-symbol and 3-symbol discretisations of the data. Analysing 1-minute and 5-minute price time series of 55 Exchange Traded Funds traded at the New York Stock Exchange, we develop a methodology to isolate true inefficiencies from other sources of regularities, such as the intraday pattern, the volatility clustering and the microstructure effects. The first two are modelled as multiplicative factors, while the microstructure is modelled as an ARMA noise process. Following an analytical and empirical combined approach, we find a strong relationship between low entropy and high relative tick size and that volatility is responsible for the largest amount of regularity, averaging 62\% of the total regularity against 18\% of the intraday pattern regularity and 20\% of the microstructure.}

%-----------------------------------------------------------
%\pagenumbering{roman}
\tableofcontents
%\pagenumbering{arabic}

%-----------------------------------------------------------
\section{Introduction}
The process of incorporating the information into prices does not occur instantaneously in real markets, giving rise to small inefficiencies, that are more present at the high frequency (intraday) level than at the low frequency (at most daily) level. Since inefficiencies are always present, rather than efficiency in absolute terms it is more interesting to study the notion of \emph{relative} efficiency, i.e., the degree of efficiency of a market measured against the benchmark of an idealised perfectly efficient market.

In this paper we investigate to what extent assets depart from the idealised perfect efficiency, ranking them according to relative efficiency, and to what degree the known sources of regularities (the intraday pattern, the volatility and the microstructure) contribute to the creation of inefficiencies. As tool to measure the randomness of the time series, we employ the Shannon entropy. Since it is defined for finite-alphabet symbolic sources, we work with 2-symbol and 3-symbol discretisations of the data.

In the 2-symbol discretisation (one symbol for the positive returns, the other for the negative returns) the intraday pattern and the volatility have no effect, since they are modelled as multiplicative factors. Only the microstructure effects affect the symbolisation. Market microstructure gives rise to linear autocorrelation that have been modelled in the literature as autoregressive moving average (ARMA) processes. By means of an analytical study on the entropy of AR(1) and MA(1) processes, and a more empirical study on ARMA residuals, we develop two methodologies to assess the inefficiency beyond the microstructure effects. Relying on the ARMA modelling of asset returns, we follow two ways. For the series that are best fitted with an AR(1), an MA(1) or a simple white noise process, we define a measure of inefficiency as the (normalised) difference between the Shannon entropy measured on the data and the theoretical value of the Shannon entropy of the corresponding process. For the other series, which are best fitted with ARMA$(p,q)$ models with $p + q > 1$, we define a measure of inefficiency as the (normalised) difference between 1 (the entropy of a white noise process) and the Shannon entropy of the ARMA residuals.

In the 3-symbol discretisation, one symbol represents a stability basin, encoding all returns in a neighbourhood of zero, while negative and positive returns lying outside of this basin are encoded with the other two symbols. Previous works dealing with a 3-symbol discretisation of data fix absolute thresholds to define this basin. We argue that there are numerous problems in doing so and, as a major enhancement with respect to such literature, we propose a very flexible approach to the 3-symbol discretisation, which is also rather general. We define the thresholds for the symbolisation to be the two tertiles of the distribution of values taken by the time series. Such a definition has many advantages, since, unlike a fixed-threshold symbolisation scheme, it adapts to the distribution of the time series. This is important because different assets have different distributions of returns and fixed thresholds could introduce discrepancies in treating the different time series. Moreover, the distribution of returns also varies with the sampling frequency, so that a fixed symbolisation scheme for different frequencies appears inappropriate. Finally, our flexible tertile symbolisation can be applied not only to the raw return series, but also to series of processed returns whose values range on different scales, such as the volatility-standardised returns and the ARMA residuals. Using the tertile symbolisation, we investigate to what degree the intraday pattern, the volatility and the microstructure contribute to create regularities in the return time series. We follow a whitening procedure, starting from the price returns, removing the intraday pattern, then standardising by the volatility, finally filtering the standardised returns for the ARMA structure and getting the ARMA residuals. We symbolise all these series with the dynamic tertile thresholds and estimate the Shannon entropy of the symbolis series to measure their degree of randomness.

In the literature only few papers have studied the relative efficiency of financial markets and ranked assets according to the efficiency degree. In \cite{Cajueiro_Tabak:2004, Giglio_etal:2008, Shmilovici_etal:2003, Shmilovici_etal:2009, Risso:2009, Oh_Kim_Eom:2007} different tools were used to measure it as distance to perfect randomness, such as the Hurst exponent, algorithmic complexity theory, a variable order Markov model, Shannon entropy and Approximate Entropy. The great majority of this literature analyses daily data (an exception is \cite{Shmilovici_etal:2009}), while our study is focused on intraday high frequency data and their peculiarities. In \cite{Giglio_etal:2008, Shmilovici_etal:2003, Shmilovici_etal:2009} return data are symbolised with a three-symbol discretisation with absolute thresholds, which in our opinion introduces some redundancy due to the long-memory properties of the volatility. A ternary discretisation with fixed thresholds incorporates some predictability in periods of high or low volatility, even under the null assumption of no correlation among the returns. When we deal with ternary discretisations of high frequency returns in Section \ref{sec:ternary_alphabet}, we will pay great attention to the removal of intraday patterns and long-term volatility.

The paper is organised as follows. In Section \ref{sec:Shannon_entropy} we introduce the Shannon entropy and present the theoretical study on the entropy of the AR(1) and MA(1) processes, in Section \ref{sec:data} we present the data, Section \ref{sec:binary_alphabet} presents the analyses performed with the binary symbolisation of the data, Section \ref{sec:ternary_alphabet} reports on the analyses done with the ternary symbolisations and Section \ref{sec:conclusions} concludes. Appendix \ref{sec:entropy_estimator} details on the entropy estimator used, Appendix \ref{sec:app_AR1_MA1_entropy} contains further details on the theoretical entropy of the AR(1) and MA(1) processes and the proofs of the propositions stated in Section \ref{sec:AR1_MA1_entropy} and Appendix \ref{sec:appendix_data_cleaning} reports on the data cleaning procedures.

%-----------------------------------------------------------
\section{Shannon entropy of high frequency data} \label{sec:Shannon_entropy}

%-----------------------------------------------------------
\subsection{Shannon entropy}
In information theory, a finite-alphabet stationary information source is a sequence of random variables which take values in a finite set $A$ (the \emph{alphabet} of the source), such that the probability
\begin{equation} \label{eq:string_measure}
	\mu (X_t = a_1, \ldots, X_{t+k-1} = a_k)
\end{equation}
of receiving a given string $a_1 \ldots a_k$ is well defined for all positive integers $k$ and independent of $t$. The measures (\ref{eq:string_measure}), for all $k$ and for all possible values of $a_1, \ldots, a_k \in A^k$, completely define the source and it is thus legitimate to identify the source with its probability measure $\mu$. The \emph{Shannon entropy} of a finite-alphabet information source is a measure of the uncertainty associated with the source's outputs. If at each time instant there is high uncertainty about what the source's output will be, the Shannon entropy of the source is high. If, conversely, there is little uncertainty, maybe because the source has some regular structures that occur repeatedly, then the Shannon entropy is low. Here is the formal definition.

\begin{dfntn}[Shannon entropy of an information source] %\label{def:Shannon_entropy}
Let $X = \{ X_1,X_2,\ldots \}$ be a stationary random process with finite alphabet $A$ and measure $\mu$. The \emph{$k$-th order entropy} of $X$ is
\begin{equation} \label{eq:blockentropy}
	H_k (\mu) = H (X_1^k) = - \sum_{x_1^k \in A^k} \mu (x_1^k) \log_2 \mu (x_1^k) .
\end{equation}
The \emph{$k$-th order conditional entropy} of $X$ is
\begin{equation} \label{eq:conditionalentropy}
	h_k (\mu) = H_k (\mu) - H_{k-1} (\mu) = - \sum_{x_1^k \in A^k} \mu (x_1^k) \log_2 \mu (x_k | x_1^{k-1}) .
\end{equation}
The \emph{entropy rate} or \emph{process entropy} of $X$ is
\begin{equation} \label{eq:entropyrate}
h (\mu) = H (\{ X_k \}) = \lim_{k \to \infty} \frac{H_k (\mu)}{k} = \lim_{k \to \infty} h_k (\mu) .
\end{equation}
\end{dfntn}
In the above definitions, the convention $0 \log_2 0 = 0$ is used. For the proof that the limits in Equation (\ref{eq:entropyrate}) exist and are equal, we refer the reader to any textbook treating process entropy.

The entropy $H_k$ represents the average uncertainty on output blocks of length $k$, while the conditional entropy $h_k$ gives a measure of the average uncertainty on a single output, once we know the most recent $k-1$ outputs. Finally, the entropy rate gives an average uncertainty on a single output, in the limit when the probability structure on longer and longer sequences is taken into account. It is thus clear why Shannon entropies are measures of departure from randomness. More random means less predictable, hence higher uncertainty and higher entropy. Conversely, a less random probability structure means higher levels of predictability and less uncertainty, thus lower entropy.

When the measure of the source is not analytically known, it must be inferred from the observed frequencies of the source's outputs. The straightforward and most natural way to estimate the measures of finite strings is to count the relative number of their occurrences in a long enough output sequence, supposed to be representative of the source's measure (which, in the limit for the length of the sequence going to infinity, is true almost surely). Suppose we have a long sample sequence $x_1^n$ generated by the information source. The empirical non-overlapping $k$-block distribution, with $k \ll n$, is defined by
\begin{equation*}
	q_k (a_1^k | x_1^n) = \frac{| \{ i \in \{ 0, 1, \ldots, m-1 \} : x_{ik+1}^{ik+k} = a_1^k \} |}{m} ,
\end{equation*}
where $m = \lfloor \frac{n}{k} \rfloor$. With this empirical measure, the entropy of the source can in principle be estimated by
\begin{equation*}
	\hat{H}_k^{\textrm{naive}} = - \sum_{a_1^k \in A^k} q_k (a_1^k | x_1^n) \log_2 q_k (a_1^k | x_1^n) .
\end{equation*}
However, when the sample $x_1^n$ is not long enough with respect to the length of the blocks $k$, this estimator is strongly biased and it systematically underestimates the value of the entropy. Better estimators have been proposed in the literature, that partially correct this bias. For all the entropy estimates $\hat{H}_k$ in this paper we use an estimator proposed by Grassberger (see Appendix \ref{sec:entropy_estimator} for the definition of the estimator and \cite{Grassberger:2008} for a detailed discussion on the entropy estimation problem and the formal derivation of the estimator (\ref{eq:Grassberger_estimator})).

Throughout this paper we estimate entropies by using the empirical distributions of finite strings in symbolic samples. In particular, we are concerned with the estimation of the entropies $H_k$ and $h_k$, with $k = 1, 2, \ldots, \log_2 n$, where $n$ is the length of the series. For values of $k$ greater than $\log_2 n$ the statistics provided by the series is too poor and the entropies $H_k$ are underestimated. For the sake of uniformity in presenting the results, we actually choose a few values for the order $k$, typically 2, 3, 6, 10.

As we shall see in Section \ref{sec:binary_alphabet} (Table \ref{tab:n_symbols_2s}), there will be cases where the different symbols of the alphabet appear in the series with significantly different frequencies. Since what we want to measure are the correlations among consecutive symbols, we want to filter out the difference in the frequencies of single symbols. Put another way, we want to measure the degree of randomness of the series, given that the symbols of the alphabet appear in the series with the observed frequencies. To this aim, the entropies we shall calculate are
\begin{equation} \label{eq:entropie_riscalate}
	\tilde{H}_k = \frac{\hat{H}_k}{\hat{H}_1} \quad \text{and} \quad \tilde{h}_k = \frac{\hat{h}_k}{\hat{H}_1} .
\end{equation}

Since the Shannon entropy deals with finite-alphabet information sources, a symbolisation of the returns is needed before being able to perform any analysis of entropy estimation. To be precise, we must say that of course return values are already discrete, since prices move on a discrete grid. However, what we intend to study by means of the Shannon entropy is the degree of randomness in the time sequence of \emph{few} coarsely identified behaviours of the price. Indeed, we shall be interested only in symbolisations into 2 or 3 symbols, each representing a notable behaviour, such as ``the price goes up'', ``the price is stationary'', ``the price goes down''.

%-----------------------------------------------------------
\subsection{Modelling high frequency data} \label{sec:modelling_hfd}
Intraday return series generally show significant departure from perfect randomness and it is a stylised fact (see, for example, \cite{Taylor:2011}) that intraday returns possess some significant correlation, at least at the first lag. Two sources of this correlation are price discreteness and the bid-ask bounce in transaction prices, which shows most clearly at higher frequencies. The discrete and bouncing prices are responsible for a negative autocorrelation at the first lag.

A possible model to theoretically explain this stylised fact is the following. Market (logarithmic) prices $p_t$ are supposed to differ from the latent efficient prices $p_t^\ast$ by pricing errors $u_t$, so that it holds
\begin{equation} \label{eq:observed_price}
p_t = p_t^\ast + u_t .
\end{equation}
The market returns
\begin{equation} \label{eq:observed_return}
r_t = r_t^\ast + u_t - u_{t-1}
\end{equation}
are therefore the sum of two terms. The first is represented by the rational returns $r_t^\ast$, which are the rational response to fundamental information and are assumed to be white noise from the theory of efficient markets. Let $\sigma^2$ indicate their variance, that is, assume that $r_t^\ast \stackrel{\text{i.i.d.}}{\sim} (0,\sigma^2)$. The second component is $u_t - u_{t-1}$. Imposing different structures on $u_t$, many structural models for the microstructure effects can be recovered. In the simplest case, $\{ u_t \}$ is an i.i.d.~noise process independent of the price process. Let $\eta^2$ indicate the variance of random variables $u_t$. The observed returns process is then MA(1) with $\bbE [r_t] = 0$ and autocovariance function given by
\begin{equation*}
	\bbE [r_t r_{t-\tau}] = \left \{
	\begin{array}{ll}
		\sigma^2 + 2 \eta^2 & \quad \text{for } \tau = 0\\
		- \eta^2            & \quad \text{for } \tau = 1\\
		0                   & \quad \text{for } \tau \geq 2
	\end{array}
	\right . .
\end{equation*}
If the pricing errors $u_t$ are assumed to follow instead an AR(1) process, then the returns process is ARMA($1,1$), with a more complex autocorrelation structure.

Some empirical high frequency data indeed show a typical MA(1) structure in the autocorrelation of returns, although others do not. In many cases, returns exhibit significant autocorrelation also at lags greater than 1 (see Figure \ref{fig:autocorrelation_ETF_returns_1min}). A typical picture is one where the autocorrelation function shows an alternating sign, decreasing in absolute size as the lag gets larger (see the top right and bottom left panels of Figure \ref{fig:autocorrelation_ETF_returns_1min}; see also \cite{ait2011ultra}). In \cite{ait2011ultra}, the authors propose a simple model to capture this alternating-sign higher order dependence, which resembles an AR(1) model and is slightly more complicated. Motivated by this similarity and by the fact that the procedure in Section \ref{sec:modelling_return_series_as_ARMApq} identifies the AR(1) model as the best ARMA($p,q$) model for some return series, we regard the AR(1) model as a good compromise between effectiveness and simplicity.

For these two simple return models (the AR(1) and the MA(1)), we develop in Section \ref{sec:AR1_MA1_entropy} an analytical approach to determine the theoretical values of their Shannon entropies.

%-----------------------------------------------------------
\subsection{The entropy of the processes AR(1) and MA(1)} \label{sec:AR1_MA1_entropy}
In order to talk about the Shannon entropy of the processes AR(1) and MA(1), whose phase space is continuous, we need some kind of discretisation. Among the infinitely many possible discretisations, we choose the simplest one which is not trivial. If $\{ X_t \}_t$ is an AR(1) or an MA(1) process, we define the binary symbolisation
\begin{equation} \label{eq:binary_symbolisation}
	s_t = B (X_t) = \left \{
		\begin{array}{ll}
			0 & \quad \textrm{if } X_t < 0\\
			1 & \quad \textrm{if } X_t > 0
		\end{array} \right .
	.
\end{equation}
The symbolisation (\ref{eq:binary_symbolisation}) is almost always defined, since the case $X_t = 0$ has obviously measure zero. In probabilistic terms, there would not be any difference between the given definition and one where the equality to zero is assigned to either symbol 0 or 1. The symbolisation (\ref{eq:binary_symbolisation}) thus defines a binary process $\{ s_t \}_t$, which will be the object we shall be studying throughout this section. This finite-state process has a measure $\mu$ inherited from and depending on the original process $X_t$. When we want to specify to which process we are referring to, we shall use the notations $H_k^{AR(1)}$, $H_k^{MA(1)}$, $h_k^{AR(1)}$, $h_k^{MA(1)}$, $h^{AR(1)}$, $h^{MA(1)}$.

To calculate the Shannon entropies (\ref{eq:blockentropy}), (\ref{eq:conditionalentropy}) and (\ref{eq:entropyrate}) of the discretised AR(1) and MA(1) processes, we exploit some properties of symmetry that they possess. We formalise these properties by proving a number of results, whose statements we report in the text of this section. Their proofs, together with a more technical part about a geometric characterisation of the entropies of the AR(1) process, are reported in Appendix \ref{sec:app_AR1_MA1_entropy}. We start with a result about the parity of the entropies as functions of the autoregressive parameter $\phi$ and the moving average parameter $\theta$.
\begin{prop} \label{prop:parity_h}
The entropy $H_k$ is an even function of the parameter $\phi$ or $\theta$, for all $k = 1, 2, \ldots$. Moreover, also $h_k$, for all $k = 1, 2, \ldots$, and $h$ are even functions. In formulas, it holds
\begin{itemize}
\item[(i)]   $H_k^{AR(1)} (\phi)   = H_k^{AR(1)} (- \phi)   ,$
\item[(ii)]  $H_k^{MA(1)} (\theta) = H_k^{MA(1)} (- \theta) ,$
\item[(iii)] $h_k^{AR(1)} (\phi)   = h_k^{AR(1)} (- \phi)   ,$
\item[(iv)]  $h_k^{MA(1)} (\theta) = h_k^{MA(1)} (- \theta) ,$
\end{itemize}
for all $k = 1, 2, \ldots$, and
\begin{itemize}
\item[(v)]   $h^{AR(1)}   (\phi)   = h^{AR(1)}   (- \phi)   ,$
\item[(vi)]  $h^{MA(1)}   (\theta) = h^{MA(1)}   (- \theta) .$
\end{itemize}
\end{prop}

Proposition \ref{prop:parity_h} tells us that it suffices to calculate the entropies of the AR(1) and the MA(1) processes only for $\phi \geq 0$ and for $\theta \geq 0$. We now state a result which derives from the symmetry of the normal distribution.
\begin{prop} \label{prop:complement_string}
Let $\mu$ be the measure of an AR(1) or an MA(1) process discretised as in (\ref{eq:binary_symbolisation}). Let $s_1^k \in \{ 0,1 \}^k$ be a binary string of length $k$ and $\bar{s}_1^k$ its \emph{complementary} string defined by $\bar{s}_i = 1 - s_i$, for each $i = 1, \ldots, k$. Then it holds $\mu (s) = \mu (\bar{s})$.
\end{prop}

Finally, the last property that we need is the time-reversibility of stationary Gaussian linear models, which is the content of the next theorem. We first give the formal definition of time-reversibility.
\begin{dfntn}
A stationary process is time-reversible if, for every $n$ and every $t_1,\ldots,t_n$, the vectors $\{ X_{t_1}, \ldots, X_{t_n} \}$ and $\{ X_{t_n}, \ldots, X_{t_1} \}$ have the same joint probability distribution.
\end{dfntn}
\begin{thrm}
Stationary ARMA processes built from a Gaussian white noise are time-reversible.
\end{thrm}
For the proof see \cite{Weiss:1975}. What we are interested in is the following specification to the AR(1) and MA(1) cases.
\begin{crllr} \label{crllr:time-reversibility}
Let $\mu$ be the measure of an AR(1) or an MA(1) process discretised as in (\ref{eq:binary_symbolisation}). Then for every binary string $s_1 \ldots s_k \in \{ 0,1 \}^k$ it holds $\mu (s_1 \ldots s_k) = \mu (s_k \ldots s_1)$.
\end{crllr}

At least for the AR(1) process, the characterisation given in Section \ref{sec:geometric_characterisation} of the appendix is very general. Though it can be exploited to calculate the entropies through the calculation of the solid angles for $k = 2, 3$, there seem to exist no general formula for calculating the solid angles in $\bbR^k$ determined by $k$ hyperplanes, for $k \geq 4$. We now find the entropies $H_k$, for $k = 1, 2, 3$, both for the AR(1) and the MA(1) processes. Thanks to Proposition \ref{prop:parity_h}, we can restrict our attention to values of the autoregressive parameter $\phi \geq 0$ and the moving average parameter $\theta \geq 0$.
\\

$\mathbf{k = 1}$\quad
When $k = 1$ we must find the measures $\mu (0)$ and $\mu (1)$. Since the random variables of the process (either the AR(1) or the MA(1)) have a symmetrical distribution we simply have $\mu (0) = \mu (1) = \frac{1}{2}$. Therefore, for the 1st order entropy we have $H_1 = - \mu (0) \log_2 \mu (0) - \mu (1) \log_2 \mu (1) = 1$.
\\

$\mathbf{k = 2}$\quad
We deal with the case $k = 2$ by proving a general result, which is very useful also for the case $k = 3$. To establish some notation, if $\mathbf{a} = a_1^l$ and $\bfb = b_1^m$ are two finite binary strings, let us denote by $\mathbf{a} \cdot^i \bfb$ the cylinder set defined by $\{ S_1^l = a_1^l \} \cap \{ S_{l+i+1}^{l+i+m} = b_1^m \}$. We prove the following two propositions for the processes AR(1) and MA(1).
\begin{prpstn} \label{prop:mu(a_b)_AR(1)}
Let $\mu$ be the measure of the discretised AR(1) process. Then it holds
\begin{align}
	\mu (0 \cdot^i 0) & = \frac{1}{2 \pi} \arccos (- \phi^{i+1}) , \label{eq:mu00_AR(1)}\\
	\mu (0 \cdot^i 1) & = \frac{1}{2 \pi} \arccos (\phi^{i+1}) , \label{eq:mu01_AR(1)}
\end{align}
for $i \geq 0$.
\end{prpstn}
\begin{prpstn} \label{prop:mu(a_b)_MA(1)}
Let $\mu$ be the measure of the discretised MA(1) process. Then it holds
\begin{align}
	\mu (00) & = \frac{1}{2 \pi} \arccos \Big(-\frac{\theta}{1+\theta^2} \Big) , \label{eq:mu00_MA(1)}\\
	\mu (01) & = \frac{1}{2 \pi} \arccos \Big( \frac{\theta}{1+\theta^2} \Big) , \label{eq:mu01_MA(1)}
\end{align}
and
\begin{equation} \label{eq:mu_s1.s2_MA(1)}
	\mu (\mathbf{s_1} \cdot^i \mathbf{s_2}) = \mu (\mathbf{s_1}) \mu (\mathbf{s_2}) ,
\end{equation}
for $i \geq 1$.
\end{prpstn}
Note that these results suffice for calculating the entropies $H_2$, since we just take $i = 0$ in Proposition \ref{prop:mu(a_b)_AR(1)} and, furthermore, by Proposition \ref{prop:complement_string}, for both processes AR(1) and MA(1) we have $\mu (10) = \mu (01)$ and $\mu (11) = \mu (00)$.

For the AR(1) process we thus have
\begin{equation*}
	H_2^{AR(1)} (\mu) = - 2 \mu (00) \log_2 \mu (00) - 2 \mu (01) \log_2 \mu (01) ,
\end{equation*}
where $\mu (00) = \frac{1}{2 \pi} \arccos (- \phi)$ and $\mu (01) = \frac{1}{2 \pi} \arccos (\phi)$.

For the MA(1) process we have
\begin{equation*}
	H_2^{MA(1)} (\mu) = - 2 \mu (00) \log_2 \mu (00) - 2 \mu (01) \log_2 \mu (01) ,
\end{equation*}
with $\mu (00)$ and $\mu (01)$ given by Equations (\ref{eq:mu00_MA(1)}) and (\ref{eq:mu01_MA(1)}).
\\

$\mathbf{k = 3}$\quad
We could find the quantities $\mu (s_1 s_2 s_3)$, with $s_i \in \{ 0,1 \}$ for $i = 1,2,3$, by using the formula to calculate the solid angles in $\bbR^3$ cut by the hyperplanes of which we have the equations. However, it is much simpler and much more instructive to exploit the symmetry properties affirmed by Proposition \ref{prop:complement_string} and Corollary \ref{crllr:time-reversibility}.

Initially, we let $\mu$ indicate the measure of either process, AR(1) or MA(1). By Proposition \ref{prop:complement_string} we have $\mu (100) = \mu (011)$, $\mu (101) = \mu (010)$, $\mu (110) = \mu (001)$, $\mu (111) = \mu (000)$. Furthermore, by Corollary \ref{crllr:time-reversibility} we also have $\mu (001) = \mu (100)$. The symmetries thus reduce the number of unknown quantities from $2^3 = 8$ to three. Now note that we have the following three independent relations:
\begin{align} \label{eq:mu(abc)}
	\mu (000) + \mu (001) & = \mu (00) , \nonumber\\
	\mu (010) + \mu (011) & = \mu (01) , \\
	\mu (000) + \mu (010) & = \mu (0 \cdot 0) . \nonumber
\end{align}
Since $\mu (00)$, $\mu (01)$, $\mu (0 \cdot 0)$ are determined for the AR(1) and MA(1) processes in propositions \ref{prop:mu(a_b)_AR(1)} and \ref{prop:mu(a_b)_MA(1)}, we can solve the system (\ref{eq:mu(abc)}).

For the AR(1) process we finally have
\begin{equation*}
	H_3^{AR(1)} (\mu) = - 2 \mu (000) \log_2 \mu (000) - 4 \mu (001) \log_2 \mu (001) - 2 \mu (010) \log_2 \mu (010) ,
\end{equation*}
where $\mu (000) = \frac{1}{2 \pi} \arccos (- \phi) - \frac{1}{4 \pi} \arccos (\phi^2)$, $\mu (001) = \frac{1}{4 \pi} \arccos (\phi^2)$, $\mu (010) = \frac{1}{2 \pi} \arccos (\phi) - \frac{1}{4 \pi} \arccos (\phi^2)$.

For the MA(1) process we have
\begin{equation*}
	H_3^{MA(1)} (\mu) = - 2 \mu (000) \log_2 \mu (000) - 4 \mu (001) \log_2 \mu (001) - 2 \mu (010) \log_2 \mu (010) ,
\end{equation*}
where $\mu (000) = \frac{1}{2 \pi} \arccos (- \frac{\theta}{1+\theta^2}) - \frac{1}{8}$, $\mu (001) = \frac{1}{8}$, $\mu (010) = \frac{1}{2 \pi} \arccos (\frac{\theta}{1+\theta^2}) - \frac{1}{8}$.
\\

\begin{figure}[h]
\includegraphics[width=0.5\textwidth]{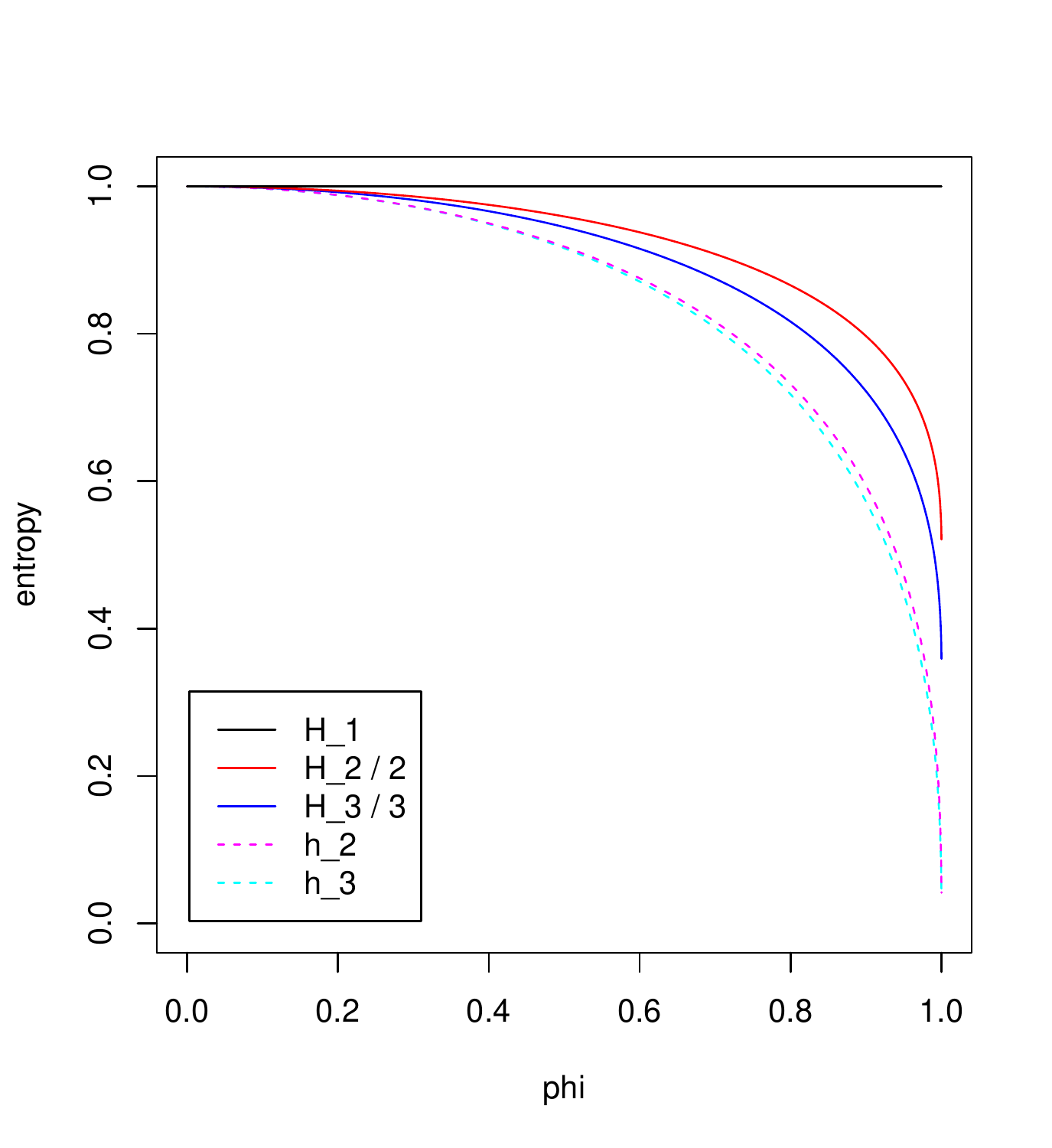}
\includegraphics[width=0.5\textwidth]{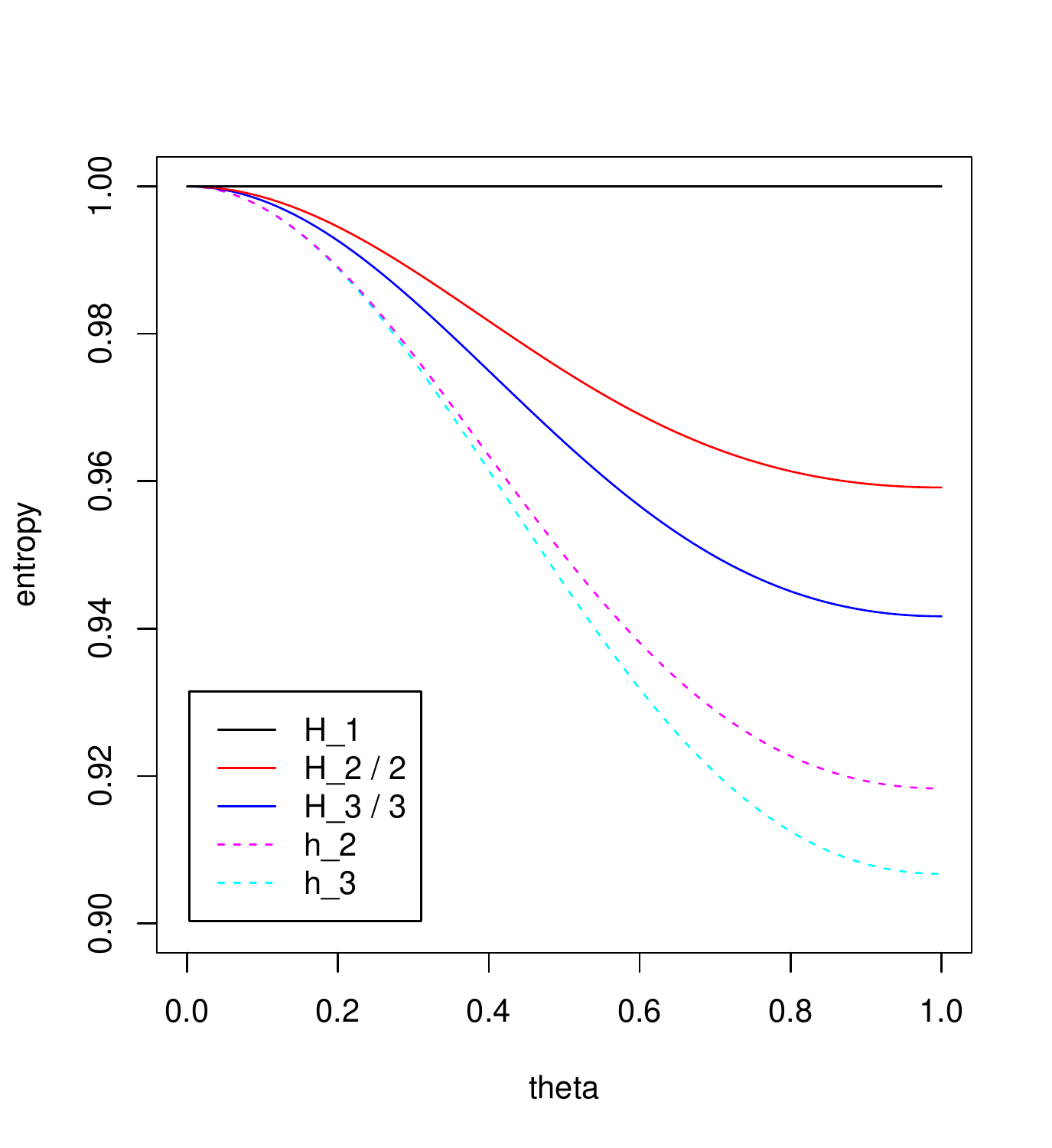}
\caption{Theoretical Shannon entropies $H_1$, $\frac{H_2}{2}$, $\frac{H_3}{3}$, together with conditional entropies $h_2 = H_2 - H_1$, $h_3 = H_3 - H_2$, of the AR(1) (left) and MA(1) (right) processes, as functions of the autoregressive parameter $0 \leq \phi < 1$ and of the moving average parameter $0 \leq \theta < 1$, respectively.}
\label{fig:entropie_teoriche}
\end{figure}

In Figure \ref{fig:entropie_teoriche} we graph the theoretical entropies $H_1$, $\frac{H_2}{2}$, $\frac{H_3}{3}$ and the conditional entropies $h_2 = H_2 - H_1$, $h_3 = H_3 - H_2$ as calculated above, for positive values of the autoregressive parameter $\phi$ and the moving average parameter $\theta$. It is interesting to note that in the AR(1) case the conditional entropies $h_2$ and $h_3$ go to 0 as $\phi$ approaches 1, while in the MA(1) process they converge to values of about 0.918 and 0.907, respectively.

%-----------------------------------------------------------
\section{Data} \label{sec:data}

%-----------------------------------------------------------
\subsection{The ETF dataset}
For our study we take high frequency historical data of 55 Exchange Traded Funds traded at the New York Stock Exchange.

An Exchange Traded Fund (ETF) is an investment whose performance is based on an index or other underlying assets. The goal of an ETF is to mimic its corresponding index and to yield the same return on investment. \emph{Inverse} and \emph{leveraged} ETFs aim at providing a return on investment which is the \emph{opposite} or a \emph{multiple} than the performance of the index. There also exist inverse leveraged ETFs.

In Table~\ref{tab:ETFs} we present the list of the ETFs studied, along with the asset tracked by each one.

\begin{table}
\centering
\begin{tiny}
\begin{tabular}{|p{0.7cm}|p{4cm}|p{1.7cm}|p{4.2cm}|c|}
\hline
ticker & name ETF & provider$^\text{a}$ & asset tracked & lev./inv.\\
% \footnotetext{St.~Str.~Gl.~Adv. = State Street Global Advisors, St.~Str.~Gl.~Mkts. = State Street Global Markets}
\hline
DIA & DIAMONDS Trust Series 1 & St.~Str.~Gl.~Adv. & Dow Jones Industrial Average Index &\\
DXD & ProShares UltraShort Dow30 & ProShares & Dow Jones Industrial Average Index & -2x\\
EEM & iShares MSCI Emerging Markets Index Fund & iShares & MSCI Emerging Markets Index &\\
EFA & iShares MSCI EAFE Index & iShares & MSCI EAFE Index &\\
EWA & iShares MSCI Australia Index & iShares & MSCI Australia Index &\\
EWG & iShares MSCI Germany Index & iShares & MSCI Germany Index &\\
EWH & iShares MSCI Hong Kong Index & iShares & MSCI Hong Kong Index &\\
EWJ & iShares MSCI Japan Index & iShares & MSCI Japan Index &\\
EWM & iShares MSCI Malaysia Index & iShares & MSCI Malaysia Index &\\
EWS & iShares MSCI Singapore Index & iShares & MSCI Singapore Index &\\
EWT & iShares MSCI Taiwan Index & iShares & MSCI Taiwan Index &\\
EWU & iShares MSCI United Kingdom Index & iShares & MSCI United Kingdom Index &\\
EWW & iShares MSCI Mexico Investable Market Index & iShares & MSCI Mexico Investable Market Index &\\
EWY & iShares MSCI South Korea Index & iShares & MSCI South Korea Index &\\
EWZ & iShares MSCI Brazil Index & iShares & MSCI Brazil Index &\\
FXI & iShares FTSE China 25 Index & iShares & FTSE China 25 Index &\\
GDX & Market Vectors Gold Miners & Van Eck & AMEX Gold Miners Index &\\
GLD & SPDR Gold Shares & St.~Str.~Gl.~Mkts. & Gold Bullion &\\
IBB & iShares Nasdaq Biotechnology Index & iShares & NASDAQ Biotechnology Index &\\
ICF & iShares Cohen \& Steers Realty Majors Index & iShares & Cohen \& Steers Realty Majors Index &\\
IJH & iShares S\&P MidCap 400 Index & iShares & S\&P MidCap 400 Index &\\
IJR & iShares S\&P SmallCap 600 Index & iShares & S\&P SmallCap 600 Index &\\
IVE & iShares S\&P 500 Value Index & iShares & S\&P 500 Value Index &\\
IVV & iShares S\&P 500 Index & iShares & S\&P 500 Index &\\
IVW & iShares S\&P 500 Growth Index & iShares & S\&P 500 Growth Index &\\
IWD & iShares Russell 1000 Value Index & iShares & Russell 1000 Value Index &\\
IWF & iShares Russell 1000 Growth Index & iShares & Russell 1000 Growth Index &\\
IWM & iShares Russell 2000 Index & iShares & Russell 2000 Index &\\
IWN & iShares Russell 2000 Value Index & iShares & Russell 2000 Value Index &\\
IWO & iShares Russell 2000 Growth Index & iShares & Russell 2000 Growth Index &\\
IYR & iShares Dow Jones U.S. Real Estate Index & iShares & Dow Jones U.S. Real Estate Index &\\
MDY & SPDR S\&P MidCap 400 & St.~Str.~Gl.~Adv. & S\&P MidCap 400 Index &\\
MZZ & ProShares UltraShort MidCap400 & ProShares & S\&P MidCap 400 Index & -2x\\
PHO & PowerShares Water Resources & PowerShares & NASDAQ OMX US Water Index &\\
QID & ProShares UltraShort QQQ & ProShares & NASDAQ-100 Index & -2x\\
QLD & ProShares Ultra QQQ & ProShares & NASDAQ-100 Index & 2x\\
QQQQ & PowerShares QQQ & PowerShares & NASDAQ-100 Index &\\
RKH & Market Vectors Bank and Brokerage & Van Eck & Market Vectors US Listed Bank and Brokerage 25 Index &\\
RTH & Market Vectors Retail & Van Eck & Market Vectors US Listed Retail 25 Index &\\
SDS & ProShares UltraShort S\&P500 & ProShares & S\&P 500 Index & -2x\\
SLV & iShares Silver Trust & iShares & price of Silver &\\
SPY & SPDR S\&P 500 & St.~Str.~Gl.~Adv. & S\&P 500 Index &\\
SSO & ProShares Ultra S\&P500 & ProShares & S\&P 500 Index & 2x\\
TIP & iShares Barclays TIPS Bond & iShares & Barclays U.S. Treasury Inflation Protected Securities Index (Series-L) &\\
USO & United States Oil & United States Commodity Funds LLP & price of West Texas Intermediate light, sweet crude oil &\\
VWO & Vanguard MSCI Emerging Markets & Vanguard & MSCI Emerging Markets Index &\\
XHB & SPDR S\&P Home Builders & St.~Str.~Gl.~Adv. & S\&P Homebuilders Select Industry Index &\\
XLB & Materials Select Sector SPDR & St.~Str.~Gl.~Adv. & Materials Select Sector Index &\\
XLE & Energy Select Sector SPDR & St.~Str.~Gl.~Adv. & Energy Select Sector Index &\\
XLF & Financial Select Sector SPDR & St.~Str.~Gl.~Adv. & Financial Select Sector Index &\\
XLI & Industrial Select Sector SPDR & St.~Str.~Gl.~Adv. & Industrial Select Sector Index &\\
XLK & Technology Select Sector SPDR & St.~Str.~Gl.~Adv. & Technology Select Sector Index &\\
XLP & Consumer Staples Select Sector SPDR & St.~Str.~Gl.~Adv. & Consumer Staples Select Sector Index &\\
XLU & Utilities Select Sector SPDR & St.~Str.~Gl.~Adv. & Utilities Select Sector Index &\\
XLV & Health Care Select Sector SPDR & St.~Str.~Gl.~Adv. & Health Care Select Sector Index &\\
XLY & Consumer Discretionary Select Sector SPDR & St.~Str.~Gl.~Adv. & Consumer Discretionary Select Sector Index &\\
\hline
\end{tabular}

\caption{List of ETFs, with provider, tracked asset and possible leverage or inverse feature. A ``2x'' leveraged ETF is one which seeks to provide 2 times the daily performance of the tracked index, with ``-2x'' standing for 2 times the inverse of the daily performance.
$^\text{a}$St.~Str.~Gl.~Adv. = State Street Global Advisors, St.~Str.~Gl.~Mkts. = State Street Global Markets.}
\label{tab:ETFs}
\end{tiny}
\end{table}

The ETFs studied include market ETFs, country ETFs, commodity ETFs and industry ETFs. The Select Sector SPDRs are ETFs that divide the S\&P 500 into nine sectors.

The data used in this study cover a period of about forty months, from the 13th July 2006 to the 1st December 2009. We use closing prices at the sampling frequency of 1 minute, and a resampling of the same data at a 5-minute frequency. In the former case we have the advantage of using the greatest possible amount of price data that is available to us, but this goes to the detriment of the regularity of the price series, since the 1-minute series present a number of missing observations depending on the level of liquidity. In the latter case, instead, we use less information but the series are more regular. So, the choice of the frequency to work with is a tradeoff between amount of information and regularity of the data. From the point of view of information theory, the best is of course to use as much data as possible, and this is what we shall predominantly do in the analyses. In order to assess to what extent results depend on the chosen frequency or on the regularity of the series, in some cases we shall perform the analyses on the 5-minute data as well. In Table \ref{tab:n_observations_ETF_data} we report the number of available price observations for the 55 ETFs. The number of days in the data sample is 854 for the vast majority of the ETFs, with six exceptions: 853 for EWU and RTH, 851 for EWW and MDY, 850 for RKH, 849 for EFA.

\begin{table}
\centering
\begin{scriptsize}
\begin{tabular}{|l|c|c|}
\hline
ETF & 1-minute observations & 5-minute observations\\
\hline
DIA & 328243 (98.81\%) & 65497 (99.60\%)\\
DXD & 276498 (83.23\%) & 63160 (96.05\%)\\
EEM & 329574 (99.21\%) & 65414 (99.48\%)\\
EFA & 327608 (99.20\%) & 65035 (99.48\%)\\
EWA & 250225 (75.32\%) & 64266 (97.73\%)\\
EWG & 192693 (58.00\%) & 61246 (93.14\%)\\
EWH & 273999 (82.48\%) & 64911 (98.71\%)\\
EWJ & 307276 (92.50\%) & 65401 (99.46\%)\\
EWM & 213888 (64.38\%) & 61752 (93.91\%)\\
EWS & 251315 (75.65\%) & 64113 (97.50\%)\\
EWT & 296362 (89.21\%) & 65143 (99.06\%)\\
EWU & 125763 (37.90\%) & 53109 (80.86\%)\\
EWW & 309977 (93.64\%) & 64921 (99.08\%)\\
EWY & 288868 (86.95\%) & 64922 (98.73\%)\\
EWZ & 328035 (98.74\%) & 65389 (99.44\%)\\
FXI & 317549 (95.59\%) & 65258 (99.24\%)\\
GDX & 280744 (84.51\%) & 64605 (98.25\%)\\
GLD & 325293 (97.92\%) & 65483 (99.58\%)\\
IBB & 218702 (65.83\%) & 62490 (95.03\%)\\
ICF & 278878 (83.95\%) & 64651 (98.32\%)\\
IJH & 208636 (62.80\%) & 62956 (95.74\%)\\
IJR & 277711 (83.60\%) & 65016 (98.87\%)\\
IVE & 224241 (67.50\%) & 63684 (96.85\%)\\
IVV & 302419 (91.03\%) & 65330 (99.35\%)\\
IVW & 239155 (71.99\%) & 64142 (97.54\%)\\
IWD & 293429 (88.33\%) & 65285 (99.28\%)\\
IWF & 300553 (90.47\%) & 65310 (99.32\%)\\
IWM & 330785 (99.57\%) & 65508 (99.62\%)\\
IWN & 291552 (87.76\%) & 65127 (99.04\%)\\
IWO & 290830 (87.55\%) & 65111 (99.02\%)\\
IYR & 313139 (94.26\%) & 65114 (99.02\%)\\
MDY & 321856 (97.23\%) & 65087 (99.33\%)\\
MZZ & 172022 (51.78\%) & 58381 (88.78\%)\\
PHO & 164836 (49.62\%) & 59896 (91.09\%)\\
QID & 324741 (97.75\%) & 65335 (99.36\%)\\
QLD & 299218 (90.07\%) & 64862 (98.64\%)\\
QQQQ & 330694 (99.54\%) & 65483 (99.58\%)\\
RKH & 246159 (74.45\%) & 61953 (94.66\%)\\
RTH & 283523 (85.45\%) & 64630 (98.40\%)\\
SDS & 302133 (90.95\%) & 64599 (98.24\%)\\
SLV & 251230 (75.62\%) & 64052 (97.41\%)\\
SPY & 330781 (99.57\%) & 65502 (99.61\%)\\
SSO & 255237 (76.83\%) & 61252 (93.15\%)\\
TIP & 190454 (57.33\%) & 61374 (93.33\%)\\
USO & 312469 (94.06\%) & 65243 (99.22\%)\\
XHB & 255791 (77.00\%) & 63412 (96.43\%)\\
XLB & 306421 (92.24\%) & 65040 (98.91\%)\\
XLE & 330192 (99.39\%) & 65404 (99.46\%)\\
XLF & 325164 (97.88\%) & 65276 (99.27\%)\\
XLI & 287595 (86.57\%) & 64975 (98.81\%)\\
XLK & 282942 (85.17\%) & 65095 (98.99\%)\\
XLP & 253028 (76.17\%) & 64589 (98.22\%)\\
XLU & 298217 (89.77\%) & 65011 (98.86\%)\\
XLV & 265267 (79.85\%) & 64688 (98.37\%)\\
XLY & 276167 (83.13\%) & 64726 (98.43\%)\\
\hline
\end{tabular}

\caption{Number of observations in the price series at frequencies of 1 minute and 5 minutes, both in absolute and relative terms.}
\label{tab:n_observations_ETF_data}
\end{scriptsize}
\end{table}

%-----------------------------------------------------------
\subsection{Data cleaning}
We now outline the steps of the data cleaning procedure and establish some terminology. We then detail on the single steps in Appendix \ref{sec:appendix_data_cleaning}.

Starting from 1-minute closing prices, we first remove outliers (see Section \ref{sec:outliers}) and then calculate logarithmic returns
\begin{equation*}
	R_t = \ln \frac{p_t}{p_{t-1}} .
\end{equation*}
We then search for possible stock splits (see Section \ref{sec:splits}), in order to remove the huge returns in correspondence of them. Throughout this paper, we will refer to returns cleaned for possible splits, but not yet processed in any other way, as \emph{raw returns}. Then, the procedure of filtering out the daily seasonalities, by the removal of the intraday pattern as explained in Section \ref{sec:intraday_pattern}, leads to what we refer to as the \emph{deseasonalised returns} $\tilde{R}_t$. Finally, we normalise the deseasonalised returns by the volatility (see Section \ref{sec:vol_proxy}), thus obtaining the \emph{standardised returns} $r_t$.

%-----------------------------------------------------------
\section{Binary alphabet} \label{sec:binary_alphabet}

%-----------------------------------------------------------
\subsection{Binary discretisation of returns} \label{sec:discretising_returns}
The simplest symbolisation of price returns is the one which distinguishes only the two cases of positive and negative return, corresponding to the two behaviours of price moving up and moving down, respectively. Stationarity of price can not be included in either of the two behaviours as long as the symbolisation is defined in a symmetrical way. This point highlights the fact that returns are distributed on a discrete set of values. If they were instead distributed on a continuous set of values, taken according to an absolutely continuous distribution, the probability of taking a precise value would be zero, thus negligible in practical cases.

If $\{ r_t \}_t$ is the time series of non-zero returns, we define the 2-symbol sequence $\{ s_t \}_t$
\begin{equation} \label{eq:ret_2s_symbolisation}
	s_t = \left \{
		\begin{array}{ll}
			0 & \quad \textrm{if } r_t < 0\\
			1 & \quad \textrm{if } r_t > 0
		\end{array} \right .
	.
\end{equation}
In Table \ref{tab:n_symbols_2s} we report the number of the two symbols `0' and `1' in the ETF return series symbolised according to symbolisation (\ref{eq:ret_2s_symbolisation}). Ten ETFs exhibit a difference in the number of the two symbols, which is statistically significant at the 1\% significance level when a null model is assumed in which at each time instants prices have the same probability of moving up and moving down. Under such assumptions, the number of either upward and downward movements of the price has a binomial distribution.

\begin{table}
\centering
\begin{scriptsize}
\begin{tabular}{|l||r|r||r|r|}
\hline
ETF & negative 1-min r& positive 1-min ret & negative 5-min ret & positive 5-min ret\\
\hline
DIA & 147229 (49.84\%) & 148163 (50.16\%) & 31193 (49.79\%) & 31457 (50.21\%)\\
DXD & 128403 (49.99\%) & 128454 (50.01\%) & 30530 (50.06\%) & 30462 (49.94\%)\\
EEM & \textbf{143965} (49.61\%) & \textbf{146212} (50.39\%) & \textbf{30645} (49.29\%) & \textbf{31525} (50.71\%)\\
EFA & \textbf{136687} (49.75\%) & \textbf{138083} (50.25\%) & 29936 (49.60\%) & 30419 (50.40\%)\\
EWA & 91641 (49.91\%) & 91962 (50.09\%) & 26932 (49.59\%) & 27376 (50.41\%)\\
EWG & \textbf{68745} (49.61\%) & \textbf{69819} (50.39\%) & \textbf{25099} (49.16\%) & \textbf{25956} (50.84\%)\\
EWH & 88689 (49.97\%) & 88792 (50.03\%) & 25504 (49.79\%) & 25720 (50.21\%)\\
EWJ & 89201 (50.07\%) & 88943 (49.93\%) & 23583 (49.86\%) & 23715 (50.14\%)\\
EWM & 59347 (49.76\%) & 59918 (50.24\%) & 21354 (49.43\%) & 21845 (50.57\%)\\
EWS & 73803 (49.92\%) & 74028 (50.08\%) & 23376 (49.60\%) & 23749 (50.40\%)\\
EWT & 89682 (49.97\%) & 89796 (50.03\%) & 24732 (49.77\%) & 24965 (50.23\%)\\
EWU & 46786 (49.93\%) & 46915 (50.07\%) & 22403 (49.88\%) & 22510 (50.12\%)\\
EWW & 127767 (50.07\%) & 127398 (49.93\%) & 29984 (49.98\%) & 30006 (50.02\%)\\
EWY & 117983 (49.87\%) & 118594 (50.13\%) & 29416 (49.57\%) & 29928 (50.43\%)\\
EWZ & \textbf{146076} (49.69\%) & \textbf{147908} (50.31\%) & \textbf{30861} (49.44\%) & \textbf{31562} (50.56\%)\\
FXI & 142538 (49.78\%) & 143791 (50.22\%) & \textbf{30939} (49.46\%) & \textbf{31609} (50.54\%)\\
GDX & 125416 (50.05\%) & 125179 (49.95\%) & 30850 (50.13\%) & 30687 (49.87\%)\\
GLD & \textbf{143987} (49.75\%) & \textbf{145424} (50.25\%) & 30773 (49.69\%) & 31152 (50.31\%)\\
IBB & 96749 (50.08\%) & 96446 (49.92\%) & 29466 (49.79\%) & 29713 (50.21\%)\\
ICF & 126453 (49.96\%) & 126646 (50.04\%) & 30910 (49.97\%) & 30949 (50.03\%)\\
IJH & 95337 (49.74\%) & 96324 (50.26\%) & 30047 (49.65\%) & 30465 (50.35\%)\\
IJR & 121864 (49.85\%) & 122588 (50.15\%) & 30438 (49.60\%) & 30928 (50.40\%)\\
IVE & 98649 (49.99\%) & 98680 (50.01\%) & 29731 (49.89\%) & 29867 (50.11\%)\\
IVV & 137548 (49.92\%) & 138016 (50.08\%) & 31272 (49.84\%) & 31478 (50.16\%)\\
IVW & 102972 (49.83\%) & 103664 (50.17\%) & 29584 (49.80\%) & 29821 (50.20\%)\\
IWD & 127387 (50.00\%) & 127382 (50.00\%) & 30500 (49.90\%) & 30623 (50.10\%)\\
IWF & 123566 (49.78\%) & 124655 (50.22\%) & 29797 (49.70\%) & 30158 (50.30\%)\\
IWM & 145131 (49.79\%) & 146334 (50.21\%) & 30908 (49.69\%) & 31291 (50.31\%)\\
IWN & 129162 (49.88\%) & 129777 (50.12\%) & 30716 (49.61\%) & 31193 (50.39\%)\\
IWO & \textbf{128803} (49.74\%) & \textbf{130171} (50.26\%) & 30909 (49.86\%) & 31079 (50.14\%)\\
IYR & 137814 (50.03\%) & 137646 (49.97\%) & 30887 (49.92\%) & 30990 (50.08\%)\\
MDY & 147600 (49.77\%) & 148983 (50.23\%) & \textbf{31168} (49.48\%) & \textbf{31819} (50.52\%)\\
MZZ & 81551 (50.11\%) & 81179 (49.89\%) & 28481 (50.35\%) & 28080 (49.65\%)\\
PHO & 67273 (49.72\%) & 68019 (50.28\%) & 26724 (49.47\%) & 27297 (50.53\%)\\
QID & 150071 (50.09\%) & 149533 (49.91\%) & 31609 (50.10\%) & 31487 (49.90\%)\\
QLD & \textbf{139639} (49.63\%) & \textbf{141716} (50.37\%) & 31519 (49.94\%) & 31598 (50.06\%)\\
QQQQ & 140048 (49.81\%) & 141141 (50.19\%) & 30368 (49.87\%) & 30523 (50.13\%)\\
RKH & 115708 (50.23\%) & 114637 (49.77\%) & 30159 (50.33\%) & 29769 (49.67\%)\\
RTH & 94231 (49.88\%) & 94683 (50.12\%) & 26214 (50.28\%) & 25921 (49.72\%)\\
SDS & 140510 (49.96\%) & 140740 (50.04\%) & 31313 (50.07\%) & 31230 (49.93\%)\\
SLV & \textbf{110158} (49.61\%) & \textbf{111891} (50.39\%) & \textbf{29885} (49.43\%) & \textbf{30569} (50.57\%)\\
SPY & 150017 (49.89\%) & 150698 (50.11\%) & 31376 (49.78\%) & 31649 (50.22\%)\\
SSO & 118571 (49.85\%) & 119300 (50.15\%) & 29536 (49.85\%) & 29718 (50.15\%)\\
TIP & \textbf{81407} (49.36\%) & \textbf{83517} (50.64\%) & \textbf{27780} (49.26\%) & \textbf{28612} (50.74\%)\\
USO & 139722 (49.91\%) & 140246 (50.09\%) & 30851 (49.60\%) & 31346 (50.40\%)\\
XHB & 102953 (50.25\%) & 101927 (49.75\%) & \textbf{29155} (50.59\%) & \textbf{28475} (49.41\%)\\
XLB & 121761 (49.79\%) & 122788 (50.21\%) & 29421 (49.67\%) & 29815 (50.33\%)\\
XLE & \textbf{146881} (49.66\%) & \textbf{148876} (50.34\%) & \textbf{30955} (49.48\%) & \textbf{31607} (50.52\%)\\
XLF & 124317 (50.21\%) & 123291 (49.79\%) & 28938 (50.25\%) & 28646 (49.75\%)\\
XLI & 108042 (50.02\%) & 107957 (49.98\%) & 28508 (49.95\%) & 28566 (50.05\%)\\
XLK & 103627 (49.87\%) & 104160 (50.13\%) & 27995 (49.81\%) & 28208 (50.19\%)\\
XLP & 89120 (49.97\%) & 89211 (50.03\%) & 26559 (49.62\%) & 26961 (50.38\%)\\
XLU & 111554 (49.92\%) & 111901 (50.08\%) & 28330 (49.80\%) & 28553 (50.20\%)\\
XLV & 96419 (49.91\%) & 96754 (50.09\%) & 27511 (49.65\%) & 27897 (50.35\%)\\
XLY & 102484 (50.03\%) & 102359 (49.97\%) & 28201 (49.83\%) & 28391 (50.17\%)\\
\hline
\end{tabular}
\caption{Number of symbols `0' and `1' in the ETF return series symbolised according to (\ref{eq:ret_2s_symbolisation}). Bold values indicate return series where statistically significant asymmetry is found in the number of the two symbols (at 99\% confidence level).}
\label{tab:n_symbols_2s}
\end{scriptsize}
\end{table}

The simplest null hypothesis to check consists in the returns being independent, that is, indistinguishable from strong white noise. For such a model, the Shannon entropy of the process symbolised according to (\ref{eq:ret_2s_symbolisation}) equals one, since all the strings $s_1^k$ have the same probability $\mu (s_1^k) = \prod_{i=1}^k \mu (s_i) = \frac{1}{2^k}$ and the uncertainty is maximum. We stress that the assumption of independence is not realistic, since intraday returns are known to possess some features of correlation, mainly due to microstructure effects. However, this first analysis allows to measure to what degree the ETF series of symbolised non-zero returns depart from complete randomness. To put it another way, it provides an average quantification of how much the sign of the returns is predictable.

We test this hypothesis by comparing the entropy of the ETF return series and the entropy of a white noise process with independent Gaussian innovations. The former is of course the entropy of a single realisation (the time series), while the latter is the entropy estimated on a Monte Carlo simulation of 1000 realisations, each of which having the same length of the ETF time series to be compared with. Indicating with $\hat{h}_k$ the estimators of the rescaled entropies $\tilde{h}_k$ of Equation (\ref{eq:entropie_riscalate}), we check whether it holds or not that
\begin{equation} \label{eq:hETF_in_intervallo_di_confidenza}
	\hat{h}_k^\text{ETF} \in [\hat{h}_k^\text{WN,0.5\%},\hat{h}_k^\text{WN,99.5\%}] ,
\end{equation}
where $\hat{h}_k^{\text{WN,}x\text{\%}}$ denotes the $x$-th percentile of the white noise Monte Carlo simulation. It turns out that the efficiency hypothesis of the ETF return series being indistinguishable from white noise is rejected for the great majority of the ETFs, as expected. There are however few cases for which this basic test fails to reject the efficiency hypothesis. We note that there is some dependence of the results on the order $k$ of the considered entropies. For $k = 2$, property (\ref{eq:hETF_in_intervallo_di_confidenza}) does not hold for ETFs IYR, XHB, XLB; for $k = 3$ it is violated by ETFs MZZ, RKH, XLB; for $k = 6$ it does not hold only for RKH; for $k = 10$ exceptions to (\ref{eq:hETF_in_intervallo_di_confidenza}) are EWZ, IBB, IYR, MZZ, XLE. In Figure \ref{fig:h10_ETF_interv_conf_rb} we show the order 10 entropies $\hat{h}_{10}^\text{ETF}$ with confidence bands of 99\%. Looking at the top picture, we see that there is a number of ETFs quite close to the efficiency level of 1, 
although just five of them are not statistically distinguishable from perfectly efficient white noise. There is however a remarkable number of ETFs which are very far from the condition of efficiency. In what follows we shall analyse to what extent this is attributable to microstructure effects that can be modelled and consequently how much of this inefficiency remains after filtering out the predictability due to microstructural dependence.
\begin{figure}[ht!]
\centering
\includegraphics[width=\textwidth]{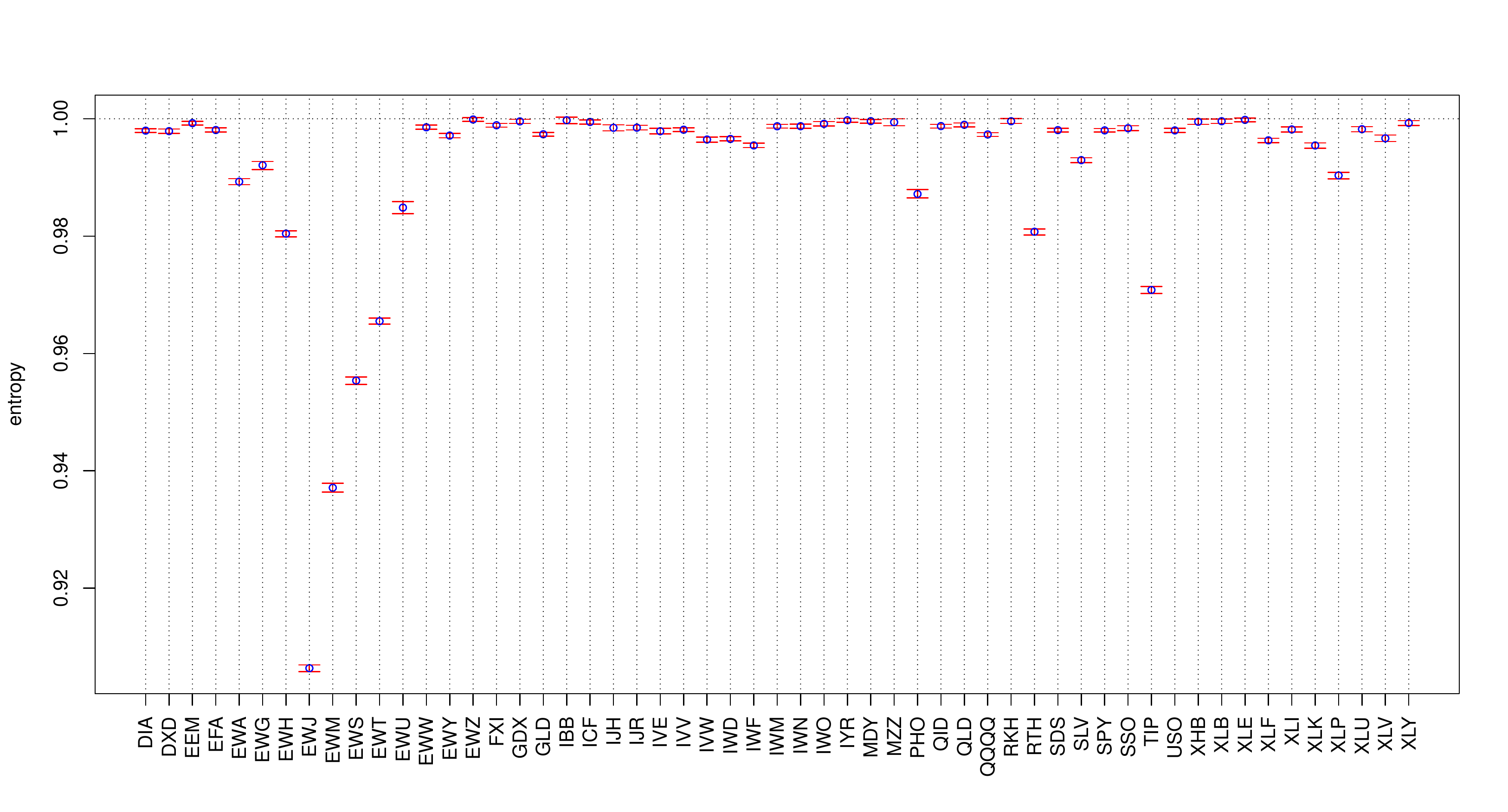}
\includegraphics[width=\textwidth]{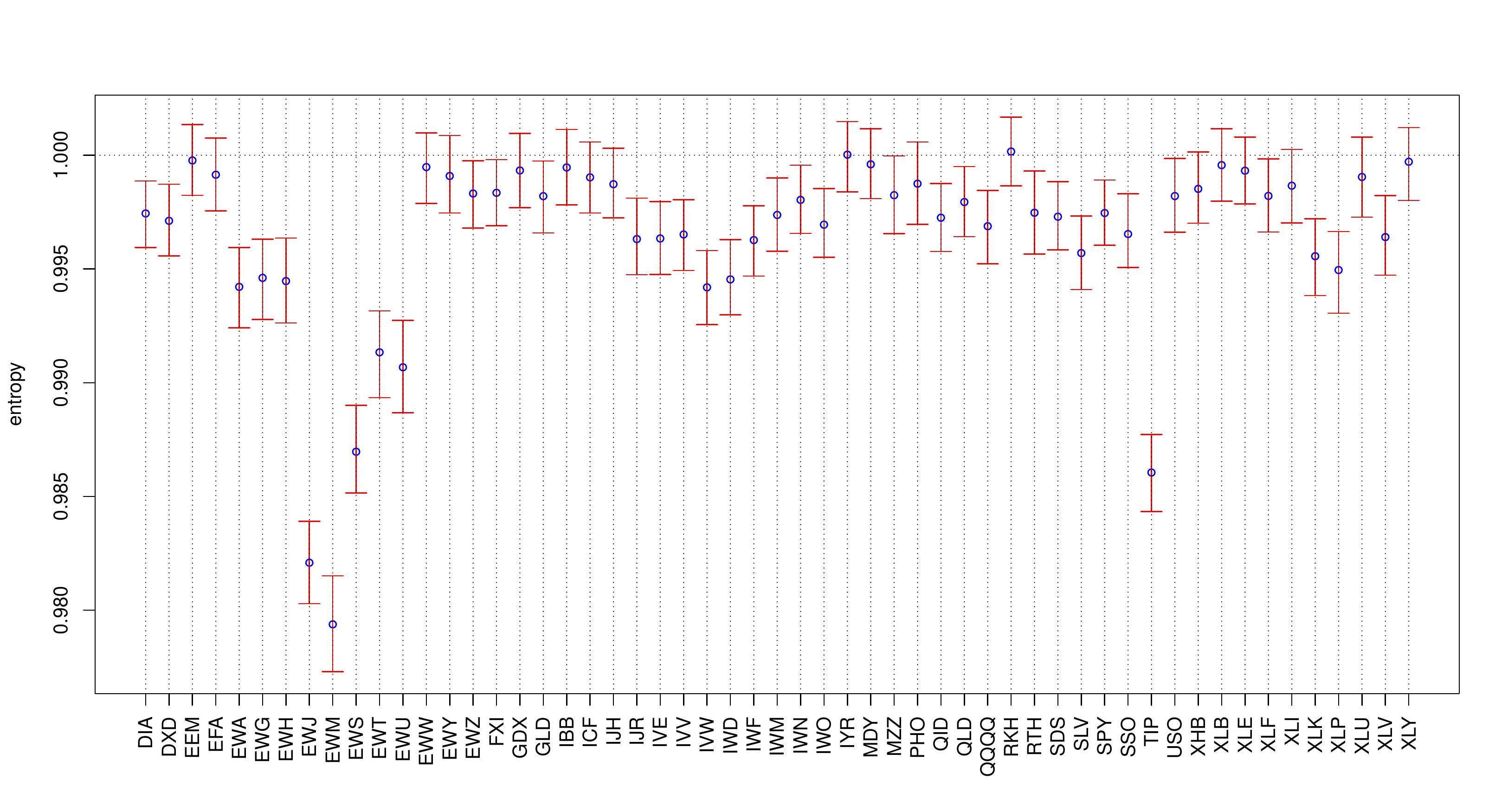}
\caption{Estimated $\hat{h}_{10}^\text{ETF}$ values (circles) of the 55 ETFs, with 99\% confidence bands obtained from Monte Carlo simulations of a white noise process of the same length of the ETFs 1-minute (top) and 5-minute (bottom) return series.}
\label{fig:h10_ETF_interv_conf_rb}
\end{figure}

By repeating the basic test (\ref{eq:hETF_in_intervallo_di_confidenza}) with the 5-minute return series (bottom graph in Figure \ref{fig:h10_ETF_interv_conf_rb}), we see how a lower sampling frequency makes the return series more efficient. Of course this is expected, since many inefficiencies must be imputable to microstructure effects, which are greater at higher frequencies. However, we still see that there are a number of ETFs (namely, EWJ, EWM, EWS, EWT, EWU, TIP) that, compared with others, maintain a large degree of inefficiency. It is also notable that in two cases (PHO and RTH) a relatively low entropy in the 1-minute return series corresponds to a relatively high entropy in the 5-minute return series. In these two cases it appears that the sole change in the frequency of the observations removes almost all the inefficiency. This contrasts with other cases of comparable low entropy for the 1-minute return series, such as EWU, which exhibit the same feature of having a relatively low entropy in the 5-minute return series as well. We point out that this should be considered as evidence of the fact that equally inefficient series may have different causes at the origins of their inefficiencies.

%---------------------------------------
\subsection{Empirical analysis of return series as AR(1) or MA(1) processes} \label{sec:returns_as_AR1_or_MA1}
As argued in Section \ref{sec:modelling_hfd}, return time series have been modelled as an AR(1) or an MA(1) process. For these two simple models, we worked out in Section \ref{sec:AR1_MA1_entropy} the analytical calculation of the entropies $H_2$, $H_3$, $h_2$, $h_3$. In this section, we take the ETFs whose non-zero return series are indeed well described by an AR(1) or an MA(1) process and we provide a quantification of their inefficiency as the degree to which the entropy measured on the empirical series differs from the theoretical one.

More precisely, our procedure is the following. We first estimate the best ARMA($p,q$) model for each one of the non-zero return series (see Section \ref{sec:modelling_return_series_as_ARMApq} for the details of the estimate procedure and the results on all the ETF series). Then, for those series that have estimated parameters $p$ and $q$ such that $p + q \leq 1$---that is, for the series whose best ARMA estimate is an AR(1), an MA(1) or just a white noise process---we can easily calculate the theoretical values $h_2^{\text{th}}$ and $h_3^{\text{th}}$ of the entropies $h_2$ and $h_3$. They just equal 1 in the case of the white noise process, while for the AR(1) and the MA(1) processes can be obtained by applying the formulas derived in Section \ref{sec:AR1_MA1_entropy}. Finally, for each of these series we compute the inefficiency scores $I_2$ and $I_3$ defined by
\begin{equation} \label{eq:ineff_score_theoretical}
	I_k = \frac{h_k^{\text{th}} - \hat{h}_k^\text{ETF}}{\sigma_k} , \qquad k = 2, 3 ,
\end{equation}
where $\hat{h}_k^\text{ETF}$ is the conditional entropy of order $k$ measured on the binarised return series and $\sigma_k$ is the standard deviation of the estimator $\hat{h}_k$ on Monte Carlo realisations of the estimated process with the same length as the ETF's return series. The scores $I_k$ measure how much the empirical series depart from being a pure AR(1), MA(1) or white noise process. Making the basic assumption that a perfectly efficient ETF should perfectly follow one of these processes (the linear ARMA dependencies being due only to microstructure features), the scores $I_k$ indeed provide a measure of the amount of inefficiency present in the empirical return series.

In Figure \ref{fig:comparison_h3_ETF_and_theoretical} we show the comparison between the entropy $\tilde{h}_3$ of the ETF return series and the theoretical value of $h_3$---given by the formulas obtained in Section \ref{sec:AR1_MA1_entropy}---of the processes AR(1), MA(1) or white noise, which best fit the return series. We notice that, as can be expected, the theoretical entropy values are always higher than the return series entropy, meaning that the return series are less efficient than the corresponding AR(1), MA(1) or white noise processes. The difference between the two values provides us with a quantification of the inefficiency of the return series against their benchmark.
\begin{figure}[h]
\centering
\includegraphics[width=\textwidth]{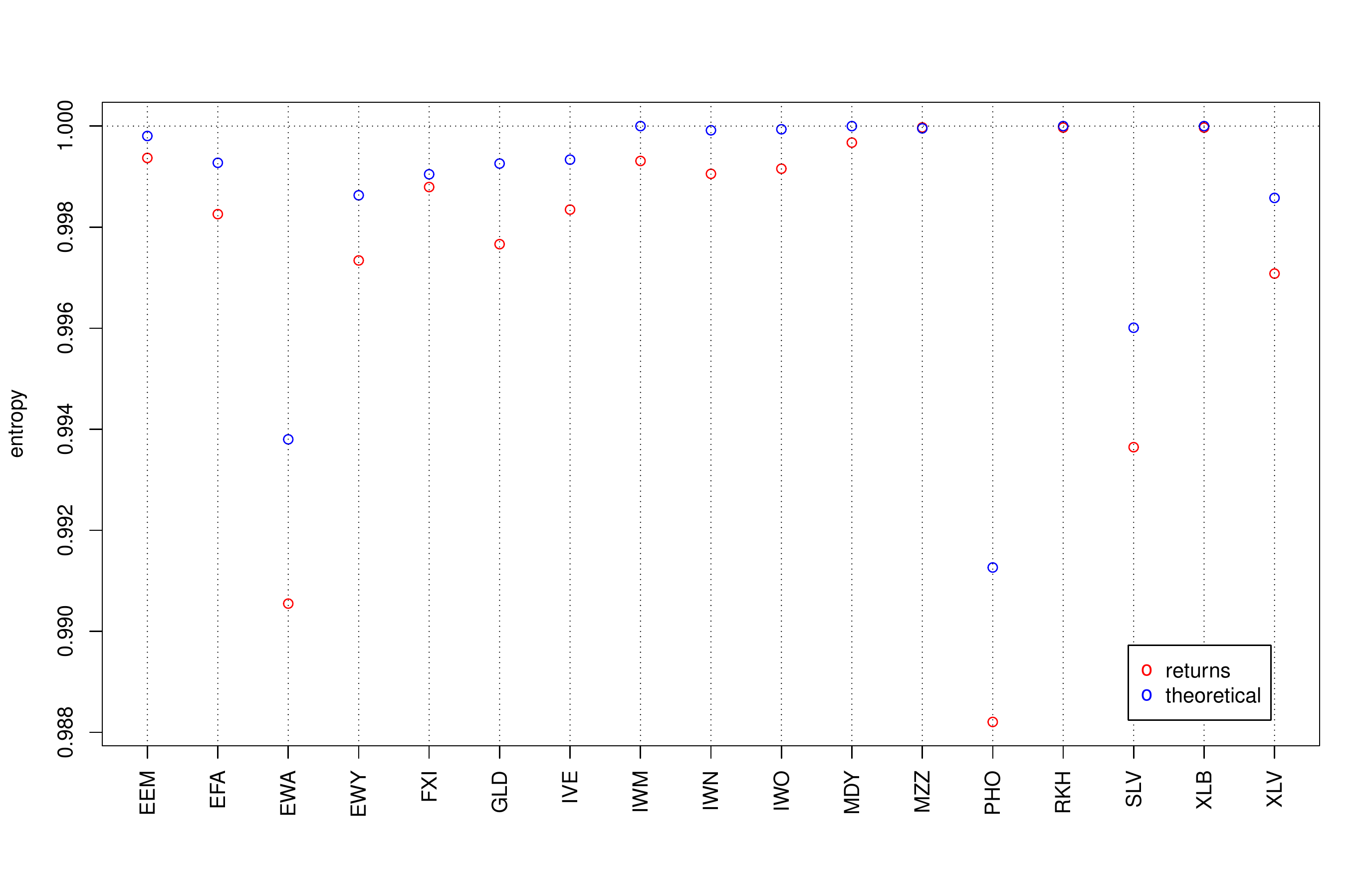}
\caption{The entropy $\tilde{h}_3$ of the non-zero 1-minute return series (red circles) and the theoretical values of the $h_3$ entropy of the corresponding best fitting AR(1), MA(1) or white noise models (blue circles), for the 17 ETFs well described by one of these models.}
\label{fig:comparison_h3_ETF_and_theoretical}
\end{figure}

In Table \ref{tab:ranking_per_entropia_teorica} we report the inefficiency rankings as determined by the inefficiency scores (\ref{eq:ineff_score_theoretical}). As the 17 ETFs treated here are a subset of the 55 ETFs for which we determine in Table \ref{tab:ranking_resARMA} other rankings, according to a different score of inefficiency, we report in parentheses this latter ranking for comparison purposes. We notice that, although the two approaches differ to a certain degree for the most inefficient ETFs (top and central rankings), they give quite stable results for the most efficient ones (bottom rankings).
\begin{table}[!htp]
\centering
\begin{scriptsize}
\begin{tabular}{|l||r|rr||r|rr|}
\hline
ETF  & $I_2$ & rank & & $I_3$ & rank & \\
\hline
IWM & 81.035 &  1 &  (9) & 65.330 &  1 &  (8)\\
IWO & 43.994 &  2 &  (4) & 38.037 &  2 &  (9)\\
MDY & 41.642 &  4 & (12) & 30.622 &  3 & (12)\\
EWA & 43.527 &  3 &  (1) & 29.653 &  4 &  (1)\\
PHO & 40.688 &  5 & (10) & 26.298 &  5 & (11)\\
IWN & 27.514 &  8 &  (8) & 20.956 &  6 &  (6)\\
GLD & 27.258 &  9 &  (5) & 19.355 &  7 &  (4)\\
XLV & 28.054 &  7 &  (7) & 19.274 &  8 &  (7)\\
SLV & 28.914 &  6 & (14) & 19.192 &  9 & (14)\\
IVE & 23.009 & 11 &  (6) & 17.939 & 10 &  (5)\\
EWY & 25.015 & 10 &  (2) & 17.012 & 11 &  (2)\\
EEM & 18.109 & 13 & (11) & 13.342 & 12 & (10)\\
EFA & 19.589 & 12 &  (3) & 13.131 & 13 &  (3)\\
FXI & 11.726 & 14 & (13) &  7.517 & 14 & (13)\\
XLB &  2.133 & 17 & (16) &  2.285 & 15 & (16)\\
RKH &  5.917 & 15 & (15) &  2.144 & 16 & (15)\\
MZZ &  2.716 & 16 & (17) &  0.360 & 17 & (17)\\
\hline
\end{tabular}
\caption{Inefficiency scores $I_k$, for $k = 2, 3$, and corresponding rankings of the 17 ETFs well described by AR(1), MA(1) or white noise model (first means most inefficient). Rankings in parentheses refer to relative positions in Table \ref{tab:ranking_resARMA}.}
\label{tab:ranking_per_entropia_teorica}
\end{scriptsize}
\end{table}

%---------------------------------------
\subsection{Empirical analysis of return series as ARMA($p,q$) processes} \label{sec:modelling_return_series_as_ARMApq}
In Section \ref{sec:modelling_hfd} we showed how a simple model for the efficient and the observed price implies that returns follow an MA(1) process. As already anticipated, the same model can explain more complex autocorrelation structures, by changing the structure of the pricing error component $u_t$ in Equations (\ref{eq:observed_price}) and (\ref{eq:observed_return}). We now change perspective and adopt an approach which is data-driven. In this section we make no a priori assumption and start from the empirical autocorrelation functions of the ETF data, which show different scenarios.

In Figure \ref{fig:autocorrelation_ETF_returns_1min} we show the sample autocorrelation of the series made up of non-zero 1-minute returns. At the top left we see a case (the ETF SLV) where significant autocorrelation is present at the lag 1 and almost at no other lag, which is the typical structure of a MA(1) process. At the top right the ETF EWA shows a situation where the autocorrelation is significant at the first lag, which is the dominant part, but also at lags 2 and 3, with an alternating sign and a decreasing absolute value which are typical of an AR(1) process with negative parameter. The ETF EWM at the bottom left of the figure has some features which recall those of EWA, such as the alternating sign and the decreasing absolute value, yet the situation is more complicated. The autocorrelation function is in fact statistically significant for many lags (from 1 to 11). Furthermore, since the decay does not look exponential, one single autoregressive parameter should not suffice to describe the 
dynamics of returns. Finally, at the bottom right we report a case (the ETF SPY) where some negative autocorrelation is significant at many lags, but there is no clear structure in the autocorrelation function. The four scenarios depicted here do not cover all the behaviours of the autocorrelation function of the 55 ETFs. However they provide a picture of how different the autocorrelation functions of different assets can be.
\begin{figure}[ht!]
\centering
\includegraphics[width=\textwidth]{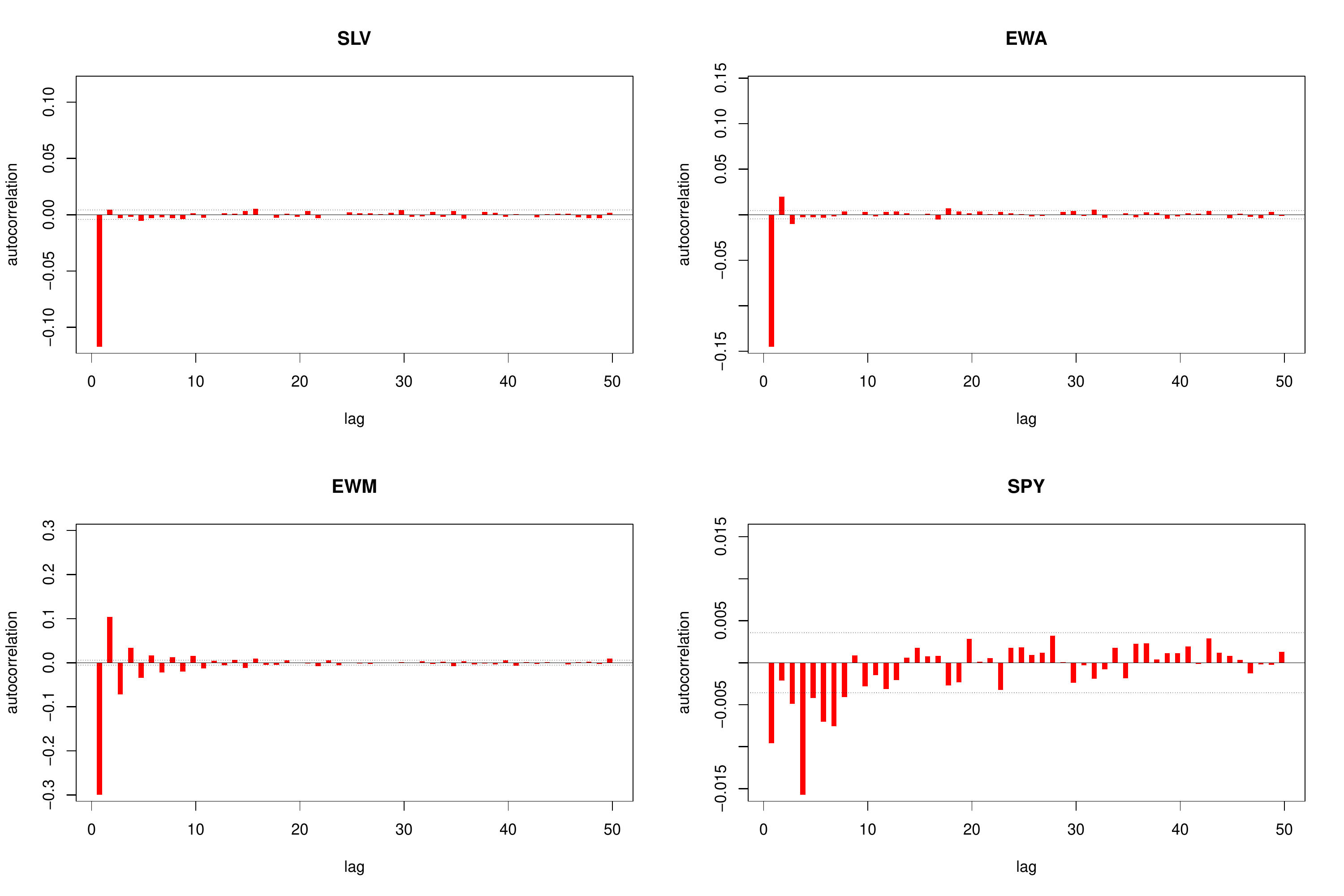}
\caption{Autocorrelation functions of non-zero 1-minute returns for the ETFs SLV, EWA, EWM, SPY. Dotted lines indicate significance levels given by $\pm \frac{1.96}{\sqrt{L}}$, where $L$ is the length of the series.}
\label{fig:autocorrelation_ETF_returns_1min}
\end{figure}

For these reasons, we do not assume any a priori model and instead we look for the best linear ARMA($p,q$) model that describes the data. In selecting the best model to fit our data we use the Bayesian Information Criterion (BIC), to strongly penalise complex models with a large number of parameters.

In Table \ref{tab:pq_ARMA_BIC} we report, for every ETF, the AR and MA orders $p$ and $q$, such that the series of non-zero 1-minute returns is best fitted with an ARMA($p,q$) model, where \emph{best} means that the model is the one that minimises the BIC. As we can see, the return series of the ETFs whose autocorrelations are represented in Figure \ref{fig:autocorrelation_ETF_returns_1min} are best fitted with different models: SLV with an MA(1) model, EWA with an AR(1) model, EWM with an ARMA($3,2$) model, SPY with an ARMA($2,4$) model.
\begin{table}[h]
\centering
\begin{scriptsize}
\begin{tabular}{|l|c|c||l|c|c||l|c|c||l|c|c||l|c|c|}
\hline
ETF & $p$ & $q$ & ETF & $p$ & $q$ & ETF & $p$ & $q$ & ETF & $p$ & $q$ & ETF & $p$ & $q$\\
\hline
DIA & 0 & 4 & EWU & 2 & 1 & IVE & 0 & 1 & PHO & 1 & 0 & USO & 3 & 2\\
DXD & 0 & 4 & EWW & 1 & 2 & IVV & 1 & 3 & QID & 2 & 2 & XHB & 2 & 0\\
EEM & 1 & 0 & EWY & 1 & 0 & IVW & 2 & 1 & QLD & 2 & 4 & XLB & 0 & 0\\
EFA & 1 & 0 & EWZ & 2 & 1 & IWD & 1 & 3 & QQQQ & 0 & 4 & XLE & 0 & 2\\
EWA & 1 & 0 & FXI & 1 & 0 & IWF & 2 & 2 & RKH & 0 & 0 & XLF & 1 & 1\\
EWG & 2 & 1 & GDX & 1 & 2 & IWM & 0 & 0 & RTH & 3 & 1 & XLI & 2 & 2\\
EWH & 1 & 2 & GLD & 1 & 0 & IWN & 1 & 0 & SDS & 1 & 3 & XLK & 2 & 2\\
EWJ & 3 & 2 & IBB & 0 & 2 & IWO & 1 & 0 & SLV & 0 & 1 & XLP & 2 & 1\\
EWM & 3 & 2 & ICF & 0 & 3 & IYR & 0 & 2 & SPY & 2 & 4 & XLU & 0 & 2\\
EWS & 3 & 2 & IJH & 0 & 4 & MDY & 0 & 0 & SSO & 1 & 1 & XLV & 1 & 0\\
EWT & 3 & 2 & IJR & 5 & 1 & MZZ & 0 & 1 & TIP & 1 & 1 & XLY & 0 & 2\\
\hline
\end{tabular}
\caption{AR and MA orders $p$ and $q$ resulting from the minimisation of the BIC among all models ARMA($p,q$) with $p + q \leq 8$.}
\label{tab:pq_ARMA_BIC}
\end{scriptsize}
\end{table}
There are 13 out of the 55 series that are best fitted either with an AR(1) model or an MA(1) model, plus other 4 that are best fitted with a simple white noise process. We stress that for these 17 cases we know what the corresponding entropies $H_2$, $H_3$, $h_2$, $h_3$ of the binary symbolisation should be, since for a white noise process they equal 1 and for the AR(1) and the MA(1) models they are given by the analytical results obtained in Section \ref{sec:AR1_MA1_entropy}. However, we now want to perform a hypothesis test to assess whether there are further dependencies other than the linear ARMA structure. If all the amount of predictability of the series is due to their linear ARMA structure, once it is filtered out there would remain no other dependence and the series of residuals should not be distinguishable from white noise. Figure \ref{fig:comparison_acf_rn_resARMA} shows how the ARMA residuals no longer contain the significant autocorrelation detected in returns.
\begin{figure}[h]
\centering
\includegraphics[width=\textwidth]{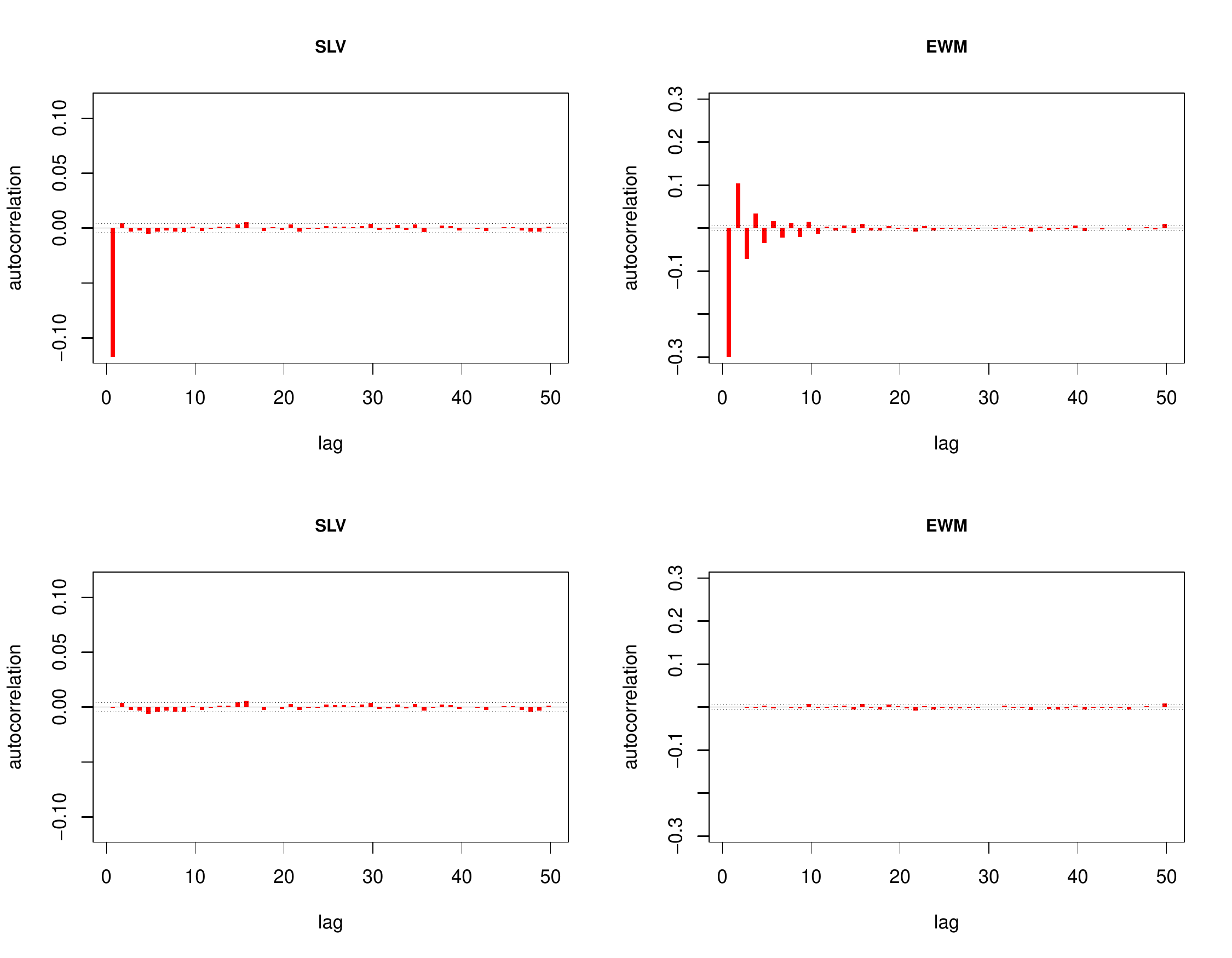}
\caption{Autocorrelation of the series of non-zero 1-minute returns (top) and of their ARMA($p,q$) residuals (bottom) for the ETFs SLV (left) and EWM (right). The values of ($p,q$) are ($0,1$) for SLV and ($3,2$) for EWM.}
\label{fig:comparison_acf_rn_resARMA}
\end{figure}
In order to measure how much inefficiency is there in the return series that is not due to the ARMA structure, we analyse the residuals of the ARMA estimates of the series, symbolise them according to (\ref{eq:ret_2s_symbolisation}) and then compare the entropy of the symbolised series with the values obtained by a Monte Carlo simulation of a white noise process. Before showing the results of this test, we first show in Figure \ref{fig:comparison_h2_ETF_and_ARMA} the comparison between the entropy $\tilde{h}_2$, defined by Equations (\ref{eq:entropie_riscalate}), of the return series and of their ARMA residuals. We stress that the same qualitative picture that we are about to describe is valid also for the entropies $\tilde{h}_k$, with $k \neq 2$. It is interesting to note how different the behaviour of the various ETFs is. Some ETFs already have a high entropy in the return series, so that taking the ARMA residuals does not lead to a noticeable increase in the entropy value. Others, that have a relatively low entropy, show a significant increase in the value of the entropy when ARMA residuals are considered in place of returns. However, there are cases where the inefficiency of the return series vanishes almost completely when taking the ARMA residuals, such as for PHO, SLV, TIP. For these ETFs it seems that a large part of the apparent inefficiency is embodied and explained by the linear dependence structure of ARMA models. For many other cases, instead, the ARMA residuals continue to contain some (nonlinear) inefficiencies.
\begin{figure}[h]
\centering
\includegraphics[width=\textwidth]{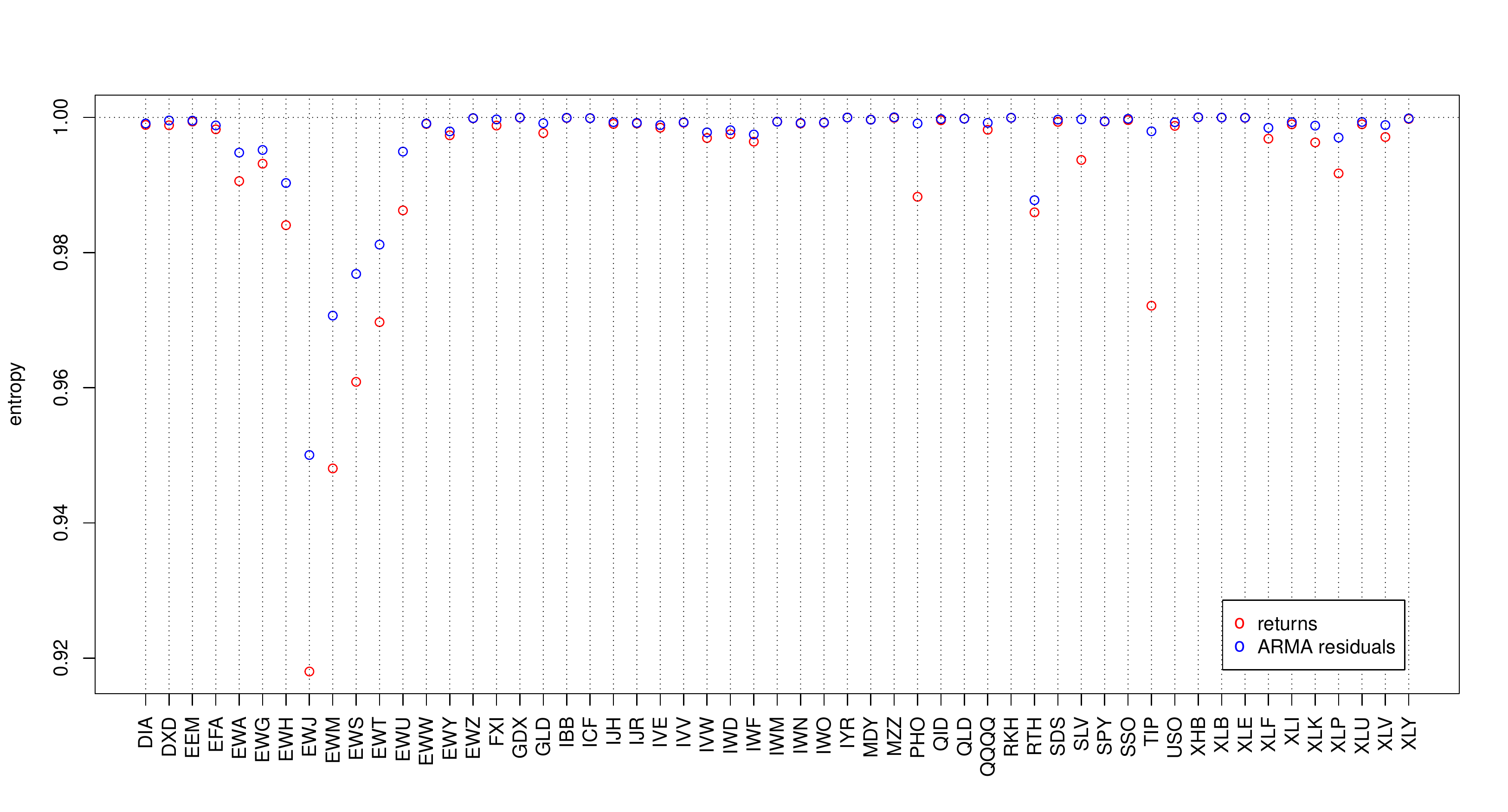}
\caption{The entropy $\tilde{h}_2$ of the non-zero 1-minute return series (red circles) and of the corresponding series of ARMA($p,q$) residuals, where $p$ and $q$ are those of Table \ref{tab:pq_ARMA_BIC} (blue circles), for the 55 ETFs.}
\label{fig:comparison_h2_ETF_and_ARMA}
\end{figure}

We now go back to the hypothesis test we mentioned earlier, to see whether the residuals of ARMA estimates are distinguishable from white noise. Analogously to what we did in Section \ref{sec:discretising_returns}, we check the condition
\begin{equation} \label{eq:hresARMA_in_intervallo_di_confidenza_WN}
	\hat{h}_k^\text{res} \in [\hat{h}_k^\text{WN,0.5\%},\hat{h}_k^\text{WN,99.5\%}] ,
\end{equation}
where $\hat{h}_k^\text{res}$ denotes the estimated entropy of the series of ARMA residuals. It turns out that even considering the ARMA residuals of return series we reject the hypothesis of efficiency for the large majority of the ETFs. The cases where the test does not reject the efficiency hypothesis are the following. For $k = 2$, GDX, MZZ, XHB, XLB; for $k = 3$, IBB, MZZ, XHB, XLB; for $k = 6$, IBB, XHB, XLE; for $k = 10$, EWZ, FXI, GDX, IBB, ICF, IYR, XHB, XLE, XLY. Referring to the results found in Section \ref{sec:discretising_returns}, we note that the cases for which the efficiency was rejected already analysing the entropy of the return series are cases where it is also rejected after filtering out the linear part, as one would expect. There are however a few exceptions (IYR for $k = 2$, RKH for $k = 3, 6$, MZZ for $k = 10$).

We argue that the predictability that remains in a series after removing the amount due to the linear component, that is, after taking the residuals of ARMA estimates, constitutes nonlinear inefficiency proper to that asset. It is interesting to quantify these inefficiencies and rank the ETFs according to their measures. We perform this by measuring how faraway the entropy of the residuals is from the value 1 of pure white noise. In particular, for each ETF we define the following measure of inefficiency:
\begin{equation}
	I_k^{\text{res}} = \frac{1 - \hat{h}_k^\text{res}}{\sigma_k^\text{WN}} ,
\end{equation}
where $\sigma_k^\text{WN}$ is the standard deviation of the estimator $\hat{h}_k$ on a white noise process of the same length as the ETF's return series. A perfectly efficient series has values of $I_k$ equal to 0. The farther from 0 the value of $I_k$, the greater the inefficiency. Results are reported in Table \ref{tab:ranking_resARMA}, for $k = 2, 3, 6, 10$. Note that varying $k$ the ranking positions are quite stable for the most inefficient ETFs (those in the first nine positions). Many other ETFs display different degrees of variability, ranging from differences across the four rankings of at most one position to differences of eleven places.
\begin{table}[!htp]
\centering
\begin{scriptsize}
\begin{tabular}{|l||r|r||r|r||r|r||r|r|}
\hline
ETF  & $I_2^{\text{res}}$ & rank & $I_3^{\text{res}}$ & rank & $I_6^{\text{res}}$ & rank & $I_{10}^{\text{res}}$ & rank\\
\hline
EWJ  & 4933.980 &  1 & 2807.921 &  1 & 1177.880 &  1 & 318.380 &  1\\
EWM  & 1505.712 &  3 & 1132.796 &  2 &  510.979 &  2 & 147.319 &  2\\
EWS  & 1495.015 &  4 & 1105.510 &  3 &  445.117 &  3 & 119.970 &  3\\
EWT  & 1911.672 &  2 & 1053.470 &  4 &  427.720 &  4 & 110.355 &  4\\
RTH  & 1307.800 &  5 &  790.725 &  5 &  328.179 &  5 &  94.070 &  5\\
EWH  &  737.441 &  6 &  546.053 &  6 &  228.858 &  6 &  70.411 &  6\\
EWA  &  417.549 &  7 &  292.831 &  7 &  118.729 &  7 &  36.094 &  7\\
IWF  &  271.103 &  9 &  207.461 &  9 &   83.385 &  8 &  23.151 &  8\\
EWG  &  371.408 &  8 &  212.230 &  8 &   81.883 &  9 &  23.147 &  9\\
IWD  &  214.192 & 13 &  166.222 & 11 &   69.487 & 10 &  20.591 & 10\\
EWU  &  208.818 & 15 &  152.508 & 14 &   55.773 & 14 &  16.896 & 11\\
IVW  &  256.876 & 10 &  155.098 & 13 &   56.317 & 13 &  16.794 & 12\\
XLP  &  232.995 & 11 &  172.684 & 10 &   66.368 & 11 &  16.750 & 13\\
EWY  &  213.702 & 14 &  156.279 & 12 &   57.306 & 12 &  15.309 & 14\\
TIP  &  187.414 & 16 &  108.947 & 16 &   48.645 & 16 &  13.820 & 15\\
QQQQ &   98.620 & 27 &   81.404 & 22 &   39.632 & 20 &  13.484 & 16\\
XLF  &  214.455 & 12 &  122.775 & 15 &   48.263 & 17 &  13.475 & 17\\
DIA  &  156.796 & 18 &  104.936 & 17 &   49.220 & 15 &  13.401 & 18\\
IVV  &  113.026 & 20 &   79.714 & 23 &   42.456 & 19 &  13.142 & 19\\
SPY  &   95.455 & 28 &   75.835 & 25 &   45.118 & 18 &  12.925 & 20\\
USO  &  106.876 & 23 &   70.114 & 26 &   34.931 & 24 &  11.419 & 21\\
GLD  &  106.747 & 24 &   81.771 & 21 &   34.481 & 25 &  10.763 & 22\\
IVE  &  100.316 & 26 &   76.238 & 24 &   35.123 & 23 &  10.693 & 23\\
IWM  &   81.035 & 31 &   65.330 & 29 &   30.048 & 27 &  10.258 & 24\\
EFA  &  186.516 & 17 &  104.336 & 18 &   38.174 & 21 &   9.839 & 25\\
XLK  &  113.568 & 19 &   89.810 & 19 &   36.301 & 22 &   9.646 & 26\\
SDS  &   52.276 & 37 &   47.967 & 35 &   29.521 & 28 &   8.911 & 27\\
IJR  &  113.007 & 21 &   63.456 & 30 &   28.002 & 29 &   8.881 & 28\\
EWW  &  102.933 & 25 &   85.422 & 20 &   33.454 & 26 &   8.665 & 29\\
XLU  &   67.310 & 33 &   48.481 & 34 &   26.354 & 32 &   7.975 & 30\\
DXD  &   52.030 & 38 &   48.874 & 33 &   24.541 & 34 &   7.344 & 31\\
IWO  &  107.853 & 22 &   62.741 & 31 &   25.647 & 33 &   7.327 & 32\\
XLV  &   95.166 & 29 &   68.339 & 28 &   27.939 & 30 &   7.320 & 33\\
IWN  &   93.928 & 30 &   68.919 & 27 &   27.043 & 31 &   7.275 & 34\\
XLI  &   66.079 & 34 &   51.481 & 32 &   21.970 & 35 &   7.262 & 35\\
QLD  &   23.058 & 43 &   23.942 & 43 &   15.057 & 40 &   7.062 & 36\\
SSO  &   21.837 & 44 &   23.412 & 44 &   16.504 & 38 &   6.430 & 37\\
IJH  &   57.385 & 36 &   43.514 & 37 &   18.616 & 36 &   6.090 & 38\\
SLV  &   28.499 & 42 &   19.818 & 45 &   10.301 & 44 &   5.525 & 39\\
QID  &   38.927 & 40 &   32.219 & 39 &   17.209 & 37 &   5.203 & 40\\
PHO  &   67.431 & 32 &   40.896 & 38 &   16.414 & 39 &   5.196 & 41\\
MDY  &   42.881 & 39 &   28.769 & 40 &   14.460 & 42 &   4.976 & 42\\
EEM  &   60.268 & 35 &   45.020 & 36 &   14.625 & 41 &   4.530 & 43\\
MZZ  &   -1.888 & 54 &    0.130 & 54 &    3.503 & 50 &   3.358 & 44\\
XLB  &    2.133 & 52 &    2.285 & 52 &    5.464 & 47 &   2.800 & 45\\
RKH  &    4.666 & 49 &    3.995 & 51 &    2.763 & 52 &   2.543 & 46\\
XHB  &   -2.395 & 55 &    0.067 & 55 &   -0.129 & 55 &   2.320 & 47\\
XLY  &   12.893 & 46 &   11.997 & 47 &   10.748 & 43 &   1.943 & 48\\
IYR  &    2.885 & 51 &    8.558 & 48 &    4.722 & 48 &   1.715 & 49\\
FXI  &   36.978 & 41 &   26.709 & 41 &    9.662 & 46 &   1.558 & 50\\
ICF  &   11.986 & 47 &   26.368 & 42 &   10.106 & 45 &   1.364 & 51\\
GDX  &    1.047 & 53 &    4.637 & 49 &    3.174 & 51 &   0.614 & 52\\
IBB  &    4.424 & 50 &    1.568 & 53 &    1.220 & 54 &  -0.372 & 53\\
EWZ  &   17.015 & 45 &   14.064 & 46 &    4.359 & 49 &  -0.747 & 54\\
XLE  &    4.929 & 48 &    4.279 & 50 &    1.440 & 53 &  -1.457 & 55\\
\hline
\end{tabular}
\caption{Inefficiency scores $I_k^{\text{res}}$, for $k = 2, 3, 6, 10$, and corresponding inefficiency rankings of the 55 ETFs (first means most inefficient).}
\label{tab:ranking_resARMA}
\end{scriptsize}
\end{table}

We do not know what the reasons of the inefficiencies and the mechanisms generating them are. They may have origin in technical details of how the trading of that particular asset is regulated by the rules of the market, or in the particular strategies adopted by market makers, or also in microstructure details (such as the relative tick size) that may have different impact on different assets. Inefficiencies may also have a more \emph{fundamental} origin, that is, they may be due to the economics of the asset and of other financial assets related to it.

Investigating the relationship between entropy and the relative tick size we find interesting results. Note that for our ETF data looking at the relative tick size (that is, the ratio between the minimum possible price variation and the price) is equivalent to looking at the price, since the absolute tick size is equal for all the 55 ETFs. It clearly emerges from the scatter plots in Figure \ref{fig:inefficiency_median_price} that the five most inefficient ETFs are those with the lowest price. A possible interpretation of this lies in the fact that an asset's price with a large relative tick size is subject to more predictable price changes. If we think of the rational price moving on a continuous scale, a change in the observed price means that the rational price has passed a tick level. Now, for an asset with a low price (or, equivalently, with a large relative tick size) it would be much more difficult to cross another tick level in the same direction than it is for an asset with a large price (that is, with a small relative tick size). Thus, most probably the observed price either moves backward or remains constant. Recall that, in the symbolisation we are considering, the stationarity of observed price (corresponding to a zero return) is ignored and simply discarded. With this considerations in mind, it is reasonable to expect the price of an asset with a large relative tick size to show some oscillating behaviour producing a more predictable symbolic sequence.
\begin{figure}[h]
\centering
\includegraphics[width=\textwidth]{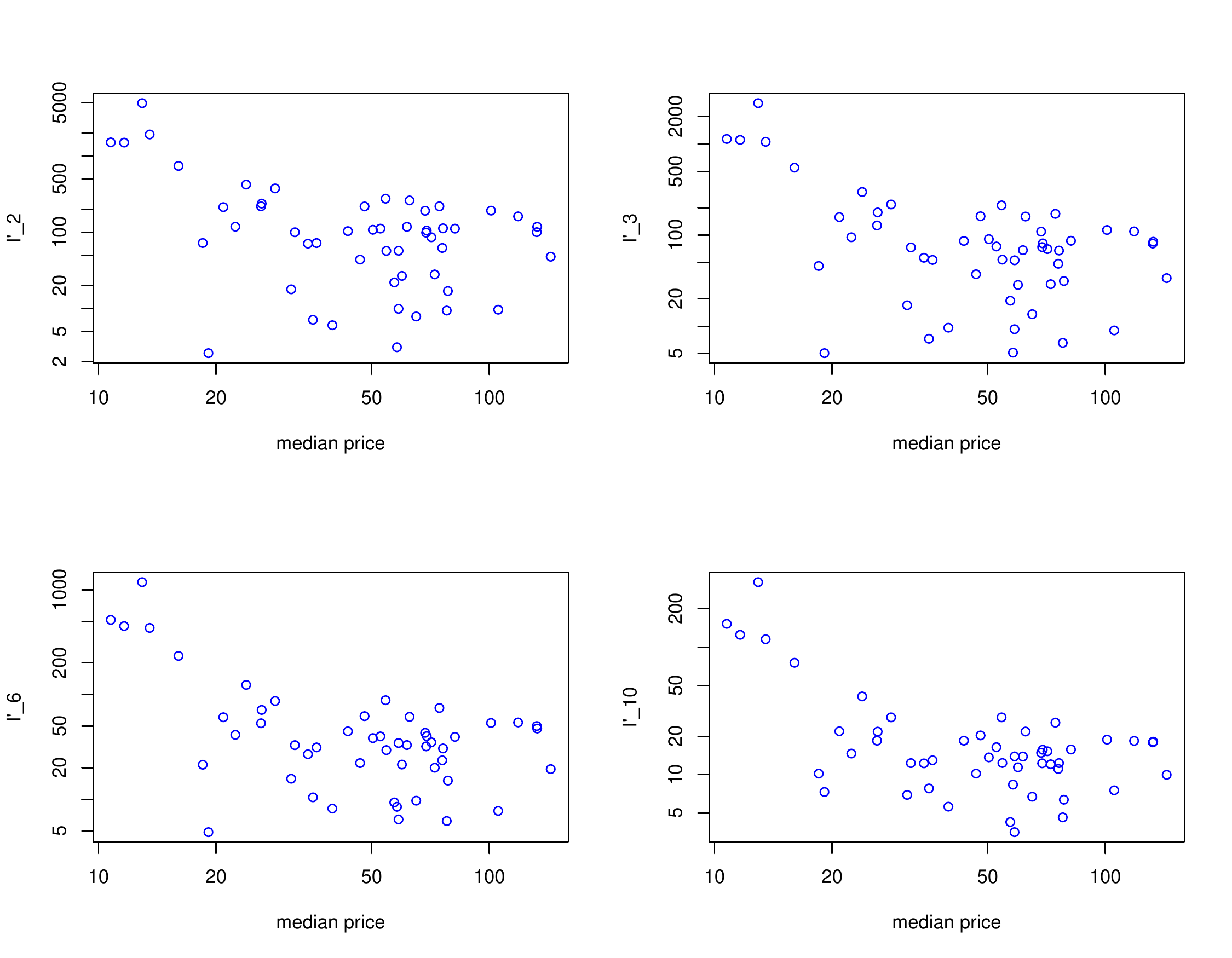}
\caption{Scatter plots in log-log scale showing the relation between inefficiency (scores $I_2^\text{res}$, $I_3^\text{res}$, $I_6^\text{res}$, $I_{10}^\text{res}$) and median price of 51 ETFs (the four ETFs with a detected split/merge are omitted). The graphed scores $I'_k = I_k^\text{res} + 5$ are shifted versions of the inefficiency scores $I_k^\text{res}$, in order to be positive for the logarithmic plot.}
\label{fig:inefficiency_median_price}
\end{figure}

Another consideration on the rankings of Table \ref{tab:ranking_resARMA} is one which concentrates on the ETFs which track a country index, in particular those from EWA to EWZ in Table \ref{tab:ETFs}. We notice that the most inefficient ETFs are those relative to the Asian countries (such as Japan, Malaysia, Singapore, Taiwan, Hong Kong) and to Australia. The ETFs tracking the indices of Germany and United Kingdom, though quite inefficient as well, are not in the very first places, lying behind the group of Asian countries. Finally, ETFs relative to Mexican and Brazilian indices are much behind in the inefficiency rankings, with the latter being among the most efficient ETFs. An exception to this classification is the ETF tracking the index of South Korea, which is not ranked in the group of ETFs of the other Asian countries. Apart from this exception, it can be argued that the levels of detected inefficiency follow the time distances from the New York time.
\begin{figure}[h]
\centering
\includegraphics[width=0.8\textwidth]{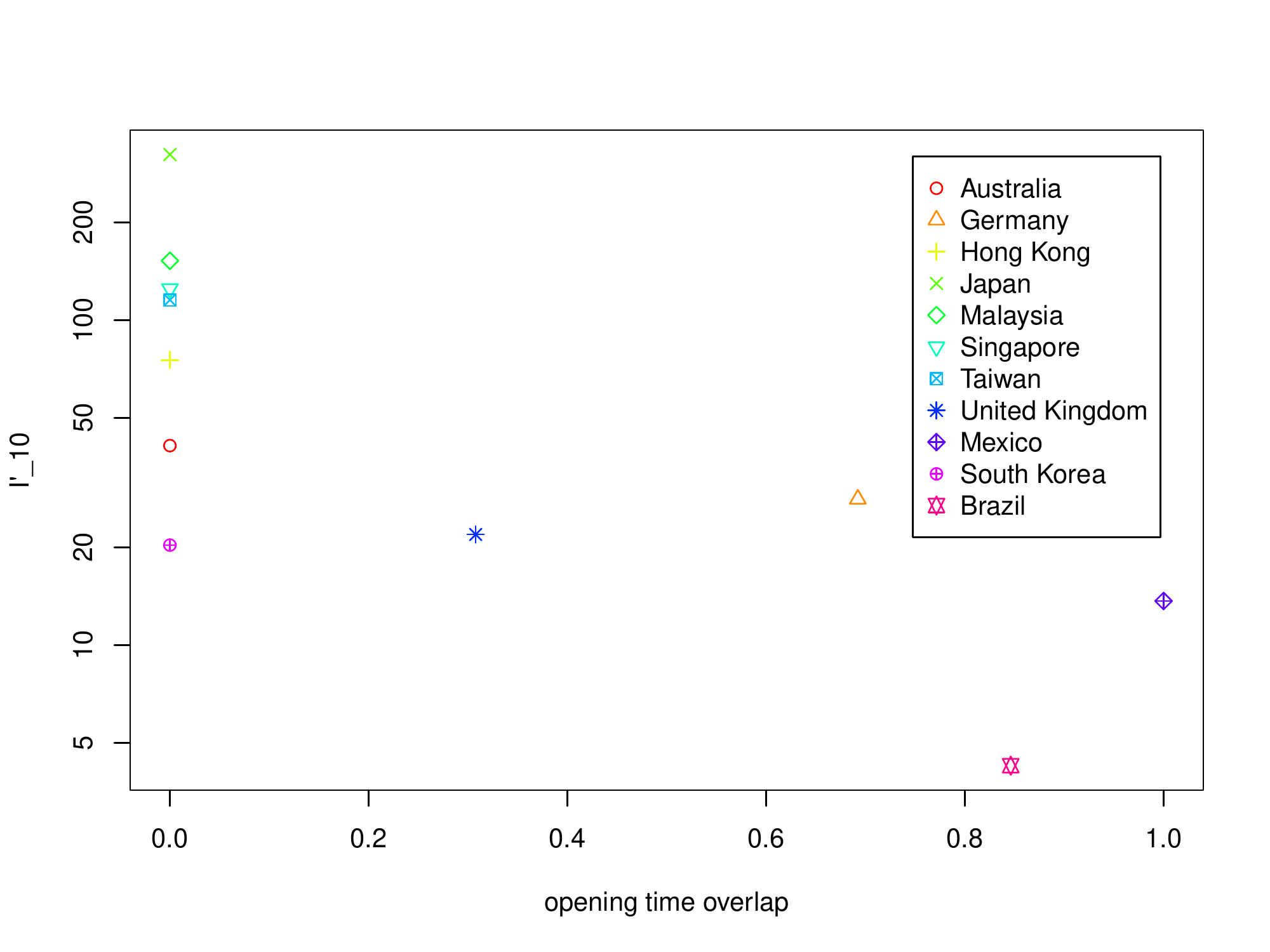}
\caption{Relationship between inefficiency and opening time overlap of the NYSE with the country markets, for the country ETFs from EWA to EWZ in Table \ref{tab:ETFs}. The graphed score $I'_{10} = I_{10}^\text{res} + 5$ is a shifted version of the inefficiency score $I_{10}^\text{res}$, in order to be positive for the logarithmic plot. The opening time overlap on the x axis is the time when both the NYSE and the country market are simultaneously open, divided by the NYSE opening time (6.5 hours).}
\label{fig:ineff-score_opening-time-overlap}
\end{figure}
In Figure \ref{fig:ineff-score_opening-time-overlap} we show the relationship between inefficiency and opening time overlap of the NYSE with the country markets, for the mentioned country ETFs. The markets of the Asian countries of Japan, Malaysia, Singapore, Taiwan, Hong Kong, as well as the Australian market, are closed when the corresponding ETFs are being traded at the New York Stock Exchange. Therefore these assets are traded while the tracked index has no dynamics. There is instead some time overlap in the opening times of the New York Stock Exchange and of European markets, while there is great overlap of the former with the markets of Mexico and Brazil. The trading dynamics of the ETFs on these last two countries can therefore rely on a simultaneous evolution of the corresponding indices.

We remark that, if the two mechanisms that we propose as possible explanations to what we observe in the inefficiency rankings are legitimate, it may well be that the two things are related. It may be that those ETFs that track indices of markets that are closed during the ETFs' trading time are deliberately given a low price, since their dynamics is only coarsely determined. However, we could not find any founded indication in this respect.

%-----------------------------------------------------------
\section{Ternary alphabet} \label{sec:ternary_alphabet}

%-----------------------------------------------------------
\subsection{Ternary discretisation of returns}
As we point out in Section \ref{sec:binary_alphabet}, a symmetrical binary discretisation of returns that takes into account all the information does not exist. The most natural one is defined by \ref{eq:ret_2s_symbolisation}, but it ignores all the zero returns, which correspond to intervals of price stationarity. This waste of information is larger as one moves to higher frequencies, since the probability of observing a price change in a fixed interval decreases with the decreasing of the interval length. In our data, at frequencies of 1 and 5 minutes the amount of zero returns is huge. In Section \ref{sec:binary_alphabet} we showed how much information on market efficiency can be extracted even ignoring the zero returns. In this section we instead use all the returns.

The idea of ternary-alphabet discretisations of returns is that a symbol represents a stability basin, encoding all returns in a neighbourhood of zero. Negative and positive returns lying outside of this basin are encoded with the two other symbols. In the papers \cite{Shmilovici_etal:2003} and \cite{Giglio_etal:2008}, the three-symbol discretisation of returns is performed according to the following definition,
\begin{equation}
	s_t = \left \{
		\begin{array}{ll}
			0 & \quad \text{if } r_t < -b\\
			1 & \quad \text{if } -b \leq r_t \leq b\\
			2 & \quad \text{if } r_t > b
		\end{array} \right .
	,
\end{equation}
with a threshold $b = 0.0025$. In \cite{Shmilovici_etal:2009}, high frequency exchange rates differences are discretised in a similar fashion with a threshold equal to three pips (i.e.~$b = 0.0003$). We argue that there are numerous problems in fixing an absolute threshold in these discretisations. The main objection is that a fixed symbolisation scheme does not take into account the heteroskedasticity of the series. Time series of returns are known to display periods of different volatility, that is, periods with different average absolute size. We believe that, in such contexts, a ternary discretisation of returns should possess a character of variability, in order to consistently adapt to the volatility dynamics. The risk in not doing so is that the memory properties of the volatility are encoded in the symbolised series and spurious dependencies are introduced.

There are also other reasons that make fixed-threshold ternary symbolisations inadequate. First of all, different assets have different distributions of returns (or rates differences in the case of currency exchanges), so that a fixed neighbourhood of zero includes portions of the return distributions which are different across the assets. This introduces discrepancies in treating the time series that can potentially affect the results of the analyses. Secondly, the distribution of returns also varies as the sampling frequency varies, so that choosing a fixed symbolisation scheme for different frequencies (as in \cite{Shmilovici_etal:2009}) appears inappropriate. Finally, the three-symbol discretisation can be applied  not only to the raw returns, but also to returns filtered for the intraday pattern as defined by ( \ref{eq:deseasonalised_return}) or to standardised returns as defined by (\ref{eq:return_standardisation}). These latter two series, as well as other possible series obtained by processing the returns in some other way, range on different scales and therefore the ternary symbolisations with fixed thresholds are not the proper way to deal with their discretisation.

Concerning the three-symbol discretisations of time series, we propose a more flexible approach, which is also rather general. We define the thresholds for the symbolisation to be the two tertiles of the distribution of values taken by the time series. More formally, if $\{ r_1,\ldots,r_N \}$ is the time series, we define its tertile-symbolised series by
\begin{equation} \label{eq:tertile_symbolisation}
	s_t = \left \{
		\begin{array}{ll}
			0 & \quad \text{if } r_t < \vartheta_1\\
			1 & \quad \text{if } \vartheta_1 \leq r_t \leq \vartheta_2\\
			2 & \quad \text{if } r_t > \vartheta_2
		\end{array} \right .
	,
\end{equation}
where $\vartheta_1$ and $\vartheta_2$ denote the two tertiles of the empirical distribution of the time series $\{ r_t \}$.

%-----------------------------------------------------------
\subsection{The impact of intraday patterns and volatility on market efficiency measures}
As already pointed out, a three-symbol discretisation of returns, with one symbol encoding returns in a stability basin around zero, embeds to some degree the intraday pattern and the volatility into the symbolic series. Thus, when these components are not properly filtered out, the symbolic series will possess a certain amount of predictability due to the memory and regularity properties of these factors. We now proceed with a quantitative study in this respect.

The whitening procedure that we apply to the series of logarithmic returns $R_t$ starts with removing the intraday pattern, getting the deseasonalised returns $\tilde{R}_t$, and continues with removing the volatility, getting the standardised returns $r_t$. We further treat the standardised returns to remove any ARMA component that may be due to microstructure factors, thus getting also a series of ARMA residuals $\epsilon_t$. We symbolise all these series with tertile thresholds as in (\ref{eq:tertile_symbolisation}) and estimate the Shannon entropy of the symbolic series to measure their degree of randomness. By doing so, we can assess to what degree the intraday pattern, the volatility and the microstructure contribute to create regularities in the return time series.

We show in Figure \ref{fig:h8_raw_deseas_stand_resid} the values of the Shannon entropy $\tilde{h}_8$ for the 1-minute and 5-minute series of raw returns $R_t$, deseasonalised returns $\tilde{R}_t$, standardised returns $r_t$ and ARMA$(\hat{p},\hat{q})$ residuals of standardised returns, where $\hat{p}$ and $\hat{q}$ are the ones corresponding to the minimum BIC value, among all the models ARMA$(p,q)$ with $p+q \leq 5$. Similar features to those discussed below also hold for results obtained with the entropies $\tilde{h}_k$, with $k \leq 10$ and $k \neq 8$.
\begin{figure}[ht!]
\centering
\includegraphics[width=\textwidth]{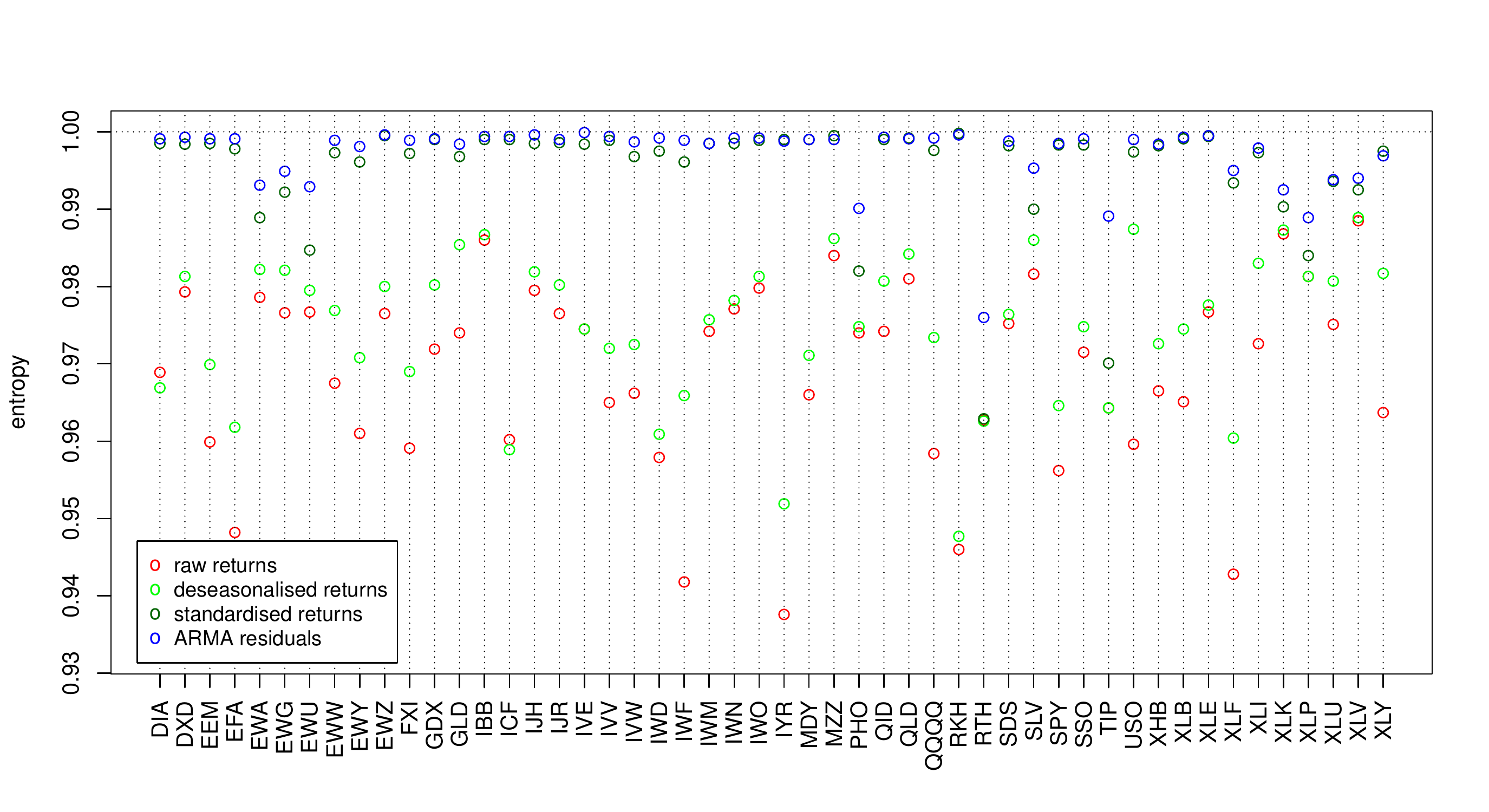}
\includegraphics[width=\textwidth]{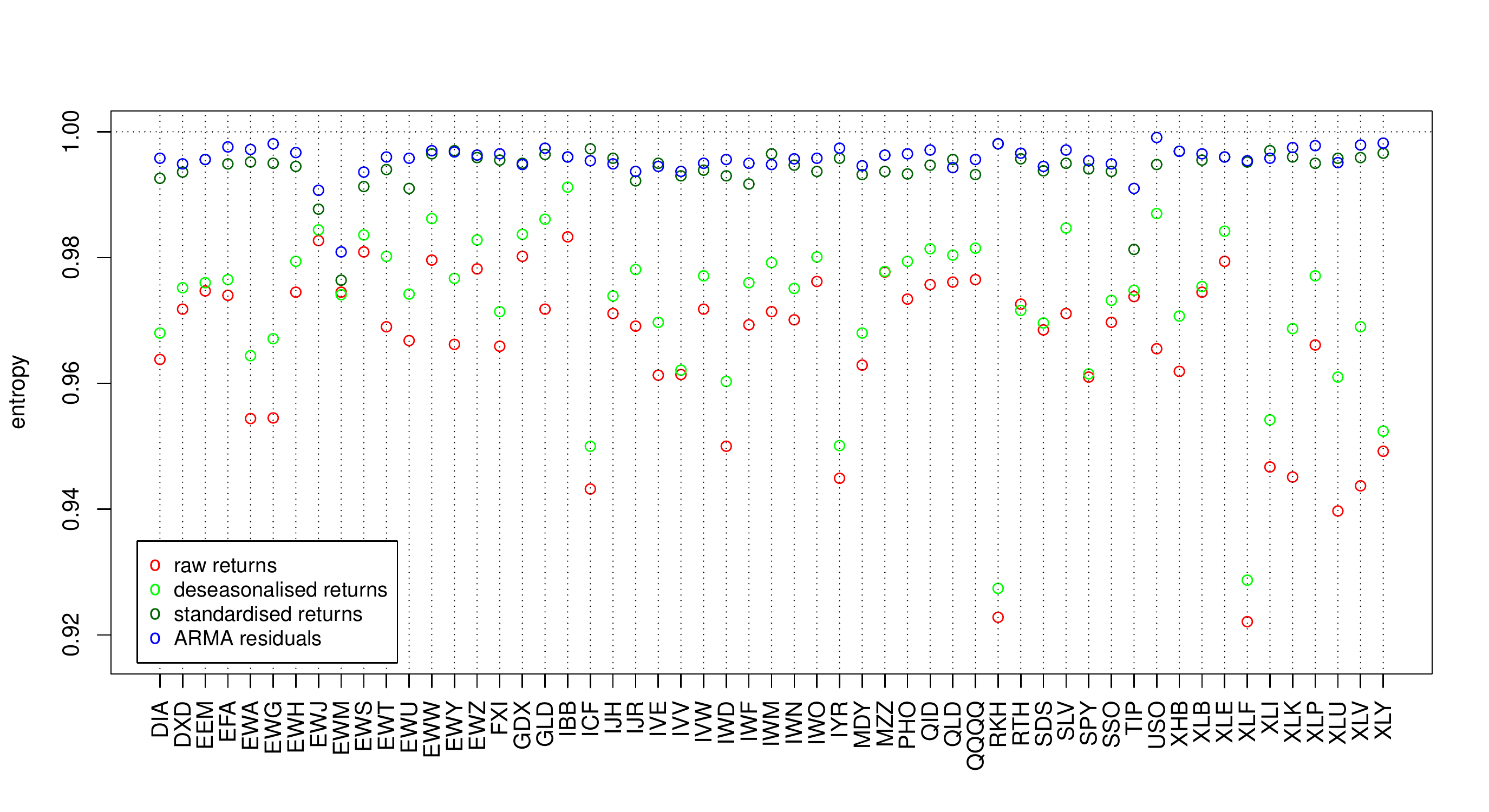}
\caption{The entropy $\tilde{h}_8$ of the 1-minute (top) and 5-minute (bottom) series of raw returns (red circles), deseasonalised returns (light green circles), standardised returns (dark green circles), ARMA residuals (blue circles), for the ETFs. Results for the ETFs EWH, EWJ, EWM, EWS, EWT are considered for the 5-minute series only.}
\label{fig:h8_raw_deseas_stand_resid}
\end{figure}

For five 1-minute series (ETFs EWH, EWJ, EWM, EWS, EWT) it happens that the proportion of zero returns is so large that the two tertile thresholds $\vartheta_1$ and $\vartheta_2$ are equal to zero. The corresponding symbolic series thus have an unbalanced number of the three symbols. In order to avoid comparing series with balanced distributions of symbols with series with unbalanced distributions, we ignore the results on the latter. We instead report results for all 5-minute series, since at this frequency the tertiles of the return distributions never happen to equal zero.

Concerning the results reported in Figure \ref{fig:h8_raw_deseas_stand_resid}, we notice that in the vast majority of cases the symbolised series of the raw returns are the most predictable, as is expected since they still carry all the regularities of the intraday pattern and the correlation due to the volatility and the microstructure. However, there are cases (DIA, ICF, RTH for the 1-minute series; EWM, RTH for the 5-minute series) in which removing the intraday pattern from the raw returns leads to series of lower entropy. This may seem to be not possible, but note that the tertile values $\vartheta_1$ and $\vartheta_2$ defining the symbolisation of the series change in an unpredictable manner when passing from the raw return series $R_t$ to the deseasonalised $\tilde{R}_t$.

The most noteworthy results of the entropy estimates reported in the two graphs of Figure \ref{fig:h8_raw_deseas_stand_resid}, however, are those regarding the standardisation of the returns by the volatility. In almost all the cases, it can be seen that the passage from the deseasonalised returns to the standardised ones is responsible for the largest increase in the entropy. Averaging across all the ETFs, a percentage of around 18\% of the entropy increase obtained with the three whitening procedures is attributable to the removal of the intraday pattern, while about 62\% is due to the standardisation for the volatility and 20\% is the entropy gain given by the removal of correlations due to microstructure effects. This means that the removal of the volatility from the return series increases their randomness more than the other two procedures taken alone. Put another way, the volatility gives the return series a huge amount of predictability and it does so much more, on average, than the daily seasonality and the microstructure effects. This result should be regarded as a convincing demonstration of the fact that, when studying the randomness of a three-symbol discretised time series, the volatility must be filtered out. Omitting this operation would give results that tell more on the predictable character of the volatility than on, for example, market efficiency.

The last refinement that we do on standardised return series is the removal of dependence structures due to market microstructure effects. As we did in Section \ref{sec:binary_alphabet} when dealing with two-symbol discretisations, we perform it by taking the residuals of the ARMA$(p,q)$ model that best describes the series of the standardised returns $r_t$. This last procedure has the general effect of removing some remaining predictability, further contributing to move return series towards perfect efficiency. We note however that this is not always the case: the entropy of the ARMA residuals is lower than that of standardised returns for the 1-minute series of IYR, MZZ, QLD, RKH, XLY and for the 5-minute series of EWY, GDX, ICF, IJH, IVE, IWM, QLD, XLI, XLU. We think that this counterintuitive behaviour might be caused by the large amount of zero returns present in the data, in correspondence of which some spurious randomness is \emph{introduced} by taking the ARMA residuals.

Overall, the results shown in Figure \ref{fig:h8_raw_deseas_stand_resid} clearly indicate that much of the apparent inefficiency (that is, the predictability of the raw return series) is due to three factors: the daily seasonality, the volatility and the microstructure effects. For example, note in the 1-minute picture of Figure \ref{fig:h8_raw_deseas_stand_resid} how four of the most apparently inefficient ETFs (namely, EFA, IWF, IYR, RKH) see their predictability almost vanish after their series are filtered for these three factors.

We remark that, although the daily seasonality and the microstructure effects are characteristic of high frequency time series, the memory properties of the volatility play a major role also in low-frequency (for example, daily) data. Therefore, we conclude that studies on measuring relative efficiency as randomness of symbolised return series should carefully deal with the issue of removing the volatility, for failing to do so would heavily affect the results, in fact invalidating them.

Our three-step procedure aims at removing from the return series all the predictability that is imputable to factors having their own dynamics that can be modelled. What remains in the filtered series that separates them from being purely random is what we assess as the true inefficiency of the assets. It may be due to other features of the market that we do not take into account or to more \emph{fundamental} aspects.

%-----------------------------------------------------------
\section{Conclusions} \label{sec:conclusions}
In this paper we study how relative market efficiency can be measured from high frequency data, filtering out all the known sources of regularity, such as the intraday pattern, the persistence of volatility and the microstructure effects. To this aim we employ the Shannon entropy as tool to measure the randomness degree of binary and ternary symbolised time series of 55 ETFs.

With an analytical study of the entropy of the AR(1) and MA(1) processes, we develop an original theoretical approach to discount microstructure effects from the measuring of efficiency of return time series which exhibit simple autocorrelation structures. A very interesting topic for future work is the extension of the analytical results found in this paper to higher entropy orders and to more complex ARMA processes.

A more empirical approach, in which we choose the ARMA($p, q$) process that best describes each time series, allows us to filter out the linear microstructure effects for all the ETFs and to measure residual regularities. Results show that in some cases a large part of the regularities is explained by the linear dependence structure, while in other cases the ARMA residuals still contain some (nonlinear) regularities. By rigorously testing the ARMA residuals for efficiency, we reject the hypothesis of efficiency for the large majority of the ETFs. We also rank the ETFs according to an inefficiency score and find that the rankings are not very sensitive to the choice of the entropy order.

We find a strong relationship between low entropy and high relative tick size. This is explained by noting that an asset's price with a large relative tick size is subject to more predictable changes. We also notice that the inefficiency scores for the country ETFs can be related to the opening time overlap between the country markets and the NYSE, where the ETFs are exchanged. We hypothesise that those ETFs that track indices of markets that are closed during the ETFs' trading time are deliberately given a low price, since their dynamics cannot rely on a simultaneous evolution of the corresponding indices.

With the 3-symbol discretisation, we find that the removal of the volatility is responsible for the largest amount of regularity in the return series. This effect amounts to the 62\% of the total entropy gain and thus it is larger, on average, than the combined effect of the intraday (18\%) and microstructure (20\%) regularity. This result convincingly demonstrates that, when studying the randomness of a three-symbol discretised time series, the volatility must be filtered out. Omitting this operation would give results that tell more on the predictable character of the volatility than on, for example, market efficiency.

\newpage
%-----------------------------------------------------------
\appendix

%-----------------------------------------------------------
\section{Entropy estimation} \label{sec:entropy_estimator}
Let us suppose to have $N$ points randomly distributed into $M$ boxes according to probabilities $p_1, \ldots, p_M$. The simplest way we can estimate the entropy $H = - \sum_{i = 1}^M p_i \log p_i$ is by replacing the $p_i$'s with the observed frequencies $\frac{n_i}{N}$, $i = 1, \ldots, M$, where $n_i$ represents the number of points in box $i$. Unfortunately, the estimator
\begin{equation*}
	\hat{H}^{\textrm{naive}} = - \sum_{i = 1}^M \frac{n_i}{N} \log \frac{n_i}{N}
\end{equation*}
is strongly biased, meaning that by employing it we would make a large systematic error in the entropy estimates. A much more accurate estimator is the one derived by Grassberger in \cite{Grassberger:2008} and defined by
\begin{equation} \label{eq:Grassberger_estimator}
	\hat{H}^{\textrm{G}} = - \sum_{i = 1}^M \frac{n_i}{N} \log \frac{\mathrm{e}^{G_{n_i}}}{N} ,
\end{equation}
where the terms $G_n$ are defined by $G_{2n+1} = G_{2n} = - \gamma - \log 2 + \frac{2}{1} + \frac{2}{3} + \frac{2}{5} + \ldots + \frac{2}{2n-1}$, with $\gamma = 0.577215\ldots$ representing Euler's constant. For the bias $\Delta \hat{H}^{\textrm{G}} = \bbE [\hat{H}^{\textrm{G}}] - H$ of this estimator, it holds
\begin{equation*}
	0 < - \Delta \hat{H}^{\textrm{G}} < 0,1407\ldots \times \frac{M}{N} .
\end{equation*}

In order to make it clear how the notations used in this section relate to the ones used in the rest of the paper, we remark that we use Grassberger's estimator to estimate the entropies of order $k$ of sources with binary or ternary alphabet, so that we typically have $M = 2^k$ or $M = 3^k$. The number $N$ and the $n_i$'s represent respectively the number of non-overlapping $k$-blocks in the observed symbolic sequence and the number of occurrences of each one of the symbolic strings of length $k$.

%-----------------------------------------------------------
\section{Details on the Shannon entropy of the processes AR(1) and MA(1)} \label{sec:app_AR1_MA1_entropy}

%---------------------------------------
\subsection{A geometric characterisation of the Shannon entropies $H_k^{AR(1)}$} \label{sec:geometric_characterisation}
In this section we give a general characterisation of the Shannon entropies $H_k^{AR(1)}$, for all $k = 1, 2, \ldots$, in terms of the entropy of some partition of the unit sphere $\bbS^{k-1} = \{ \bfx \in \bbR^k \, | \, ||\bfx|| = 1 \}$. In principle, the same path can be followed to obtain an analogous general characterisation for the entropies $H_k^{MA(1)}$. However, this does not seem to be feasible, since for the process MA(1) the general formulas for the conditional distributions of $X_k$, given $X_1,X_2,\ldots,X_{k-1}$, are not as simple as the ones for the process AR(1), which is Markov.

Let $s_1^k \in \{ 0,1 \} ^k$ be one of the $2^k$ binary strings of length $k$. According to the symbolisation (\ref{eq:binary_symbolisation}), it corresponds to the event $\{ X_1 \in I_1, \ldots , X_k \in I_k \}$, where $I_i = (-\infty,0)$ if $s_i = 0$ and $I_i = (0,\infty)$ if $s_i = 1$. For the process AR(1) we have
\begin{equation*}
\left \{
\begin{array}{lcl}
	X_1 & \! \! \! \sim & \! \! \! \calN \left( 0,\frac{\sigma^2}{1-\phi^2} \right)\\
	X_2 & \! \! \! \sim & \! \! \! \calN (\phi X_1,\sigma^2)\\
		& \! \! \! \vdots &\\
	X_k & \! \! \! \sim & \! \! \! \calN (\phi X_{k-1},\sigma^2)
\end{array}
\right .
\end{equation*}
and therefore
\begin{multline} \label{eq:probability_X_t}
\mu (s_1^k) = \int \limits_{I_1} \! \! \frac{1}{\sqrt{2 \pi} \frac{\sigma}{\sqrt{1-\phi^2}}} e^{-\frac{1}{2} \left( \frac{X_1}{\frac{\sigma}{\sqrt{1-\phi^2}}} \right)^2}
 			  \! \! \! \int \limits_{I_2} \! \! \frac{1}{\sqrt{2 \pi} \sigma} e^{-\frac{1}{2} \left( \frac{X_2 - \phi X_1}{\sigma} \right)^2}
 			  \! \ldots \\
 			  \! \ldots
 			  \! \int \limits_{I_k} \! \! \frac{1}{\sqrt{2 \pi} \sigma} e^{-\frac{1}{2} \left( \frac{X_k - \phi X_{k-1}}{\sigma} \right)^2}
 			  \dd X_k \, \ldots \, \dd X_2 \, \dd X_1 .
\end{multline}

Let us now consider the normalising linear transformation
\begin{equation} \label{eq:normalisation}
\left \{
\begin{array}{rcl}
Y_1 & = & \frac{1}{\frac{\sigma}{\sqrt{1-\phi^2}}} X_1\\
Y_2 & = & \frac{X_2 - \phi X_1}{\sigma}\\
& \vdots &\\
Y_k & = & \frac{X_k - \phi X_{k-1}}{\sigma}
\end{array}
\right . ,
\end{equation}
described in matrix form by $Y = A_\phi X$, with
\begin{equation*}
A_\phi = \frac{1}{\sigma}
\begin{pmatrix}
\sqrt{1-\phi^2} & 0 & 0 & \ldots & 0 & 0\\
-\phi & 1 & 0 & \ldots & 0 & 0\\
0 & -\phi & 1 & \ldots & 0 & 0\\
\vdots & \vdots & \vdots & \ddots & \vdots & \vdots\\
0 & 0 & 0 & \ldots & 1 & 0\\
0 & 0 & 0 & \ldots & -\phi & 1\\
\end{pmatrix} .
\end{equation*}
The random variables $Y_t$ are $\calN (0,1)$ and Equation (\ref{eq:probability_X_t}) can be written
\begin{equation} \label{eq:probability_Y_t}
\mu (s_1^k) = \int \limits_{I'} \! \! \frac{1}{(2 \pi)^\frac{k}{2}} e^{-\frac{1}{2} \left( Y_1^2 + Y_2^2 + \ldots + Y_k^2 \right)}
			  \, \dd Y_1 \, \ldots \, \dd Y_k ,
\end{equation}
where $I' = A_\phi (I_1 \times I_2 \times \ldots \times I_k)$. The integral in Equation (\ref{eq:probability_Y_t}) is equal to the fraction of $k$-dimensional solid angle determined by the cone $I'$, or, equivalently, to the fraction of hypersphere $\frac{\lambda (I' \cap \bbS^{k-1})}{\lambda (\bbS^{k-1})}$, being $\lambda$ the Lebesgue measure.

The $2^k$ solid angles of the form $I'$, corresponding to the strings of $k$ binary symbols, are those that result from sectioning the $k$-dimensional Euclidean space with the hyperplanes $\pi_1$, $\pi_2$, $\pi_3$, \ldots, $\pi_k$ of equations
\begin{align*}
\frac{\phi^{k-1}}{\sqrt{1-\phi^2}} x_1 \phantom{+ \phi^{k-2} x_2 + \phi^{k-3} x_3 + \ldots + \phi x_{k-1} + x_k} & = 0\\
\frac{\phi^{k-1}}{\sqrt{1-\phi^2}} x_1 + \phi^{k-2} x_2 \phantom{+ \phi^{k-3} x_3 + \ldots + \phi x_{k-1} + x_k} & = 0\\
\frac{\phi^{k-1}}{\sqrt{1-\phi^2}} x_1 + \phi^{k-2} x_2 + \phi^{k-3} x_3 \phantom{+ \ldots + \phi x_{k-1} + x_k} & = 0\\
\vdots \phantom{x_2 + \phi^{k-3} x_3 + \ldots + \phi x_{k-1} + x_k} & \\
\frac{\phi^{k-1}}{\sqrt{1-\phi^2}} x_1 + \phi^{k-2} x_2 + \phi^{k-3} x_3 + \ldots + \phi x_{k-1} + x_k & = 0 .
\end{align*}
The problem of calculating the measures $\mu (s_1^k)$ in Equation (\ref{eq:blockentropy}) has thus been translated into a purely geometric problem: calculating the solid angles in $\bbR^k$ cut by the hyperplanes $\pi_i$, $i = 1, \ldots, k$. The entropy of Equation (\ref{eq:blockentropy}) is thus nothing else than the entropy of the partition of $\bbS^{k-1}$ determined by the hyperplanes $\pi_i$.

%---------------------------------------
\subsection{Proofs of the propositions of Section \ref{sec:AR1_MA1_entropy}} %\label{sec:proofs}

\begin{proof}[Proof of Proposition \ref{prop:parity_h}]
First note that, if $\{ \epsilon_t \}_t$ is a Gaussian white noise, then also $\{ \epsilon_t^\prime \}_t = \{ (-1)^t \epsilon_t \}_t$ is a Gaussian white noise and it is indeed the same process as $\{ \epsilon_t \}_t$ since a Gaussian random variable $\epsilon_t$ has the same distribution as its opposite $- \epsilon_t$. The AR(1) process defined by $X_t^\prime = - \phi X_{t-1}^\prime + \epsilon_t^\prime$ has the MA($\infty$) form
\[ X_t^\prime = \sum_{i=0}^\infty (- \phi)^i \epsilon_{t-i}^\prime = \sum_{i=0}^\infty (-1)^i \phi^i (-1)^{t-i} \epsilon_{t-i} = \sum_{i=0}^\infty (-1)^t \phi^i \epsilon_{t-i} . \]
Thus we have $X_t^\prime = (-1)^t X_t$, for all $t$. This relation between the two continuous-state processes $\{ X_t^\prime \}$ and $\{ X_t \}$ translates into an analogous one for the binary processes $S^\prime = \{ s_t^\prime \}_t$ and $S = \{ s_t \}_t$ defined by discretisation as in (\ref{eq:binary_symbolisation}). This means that a single realisation of the process $\{ X_t \}_t$ (or, equivalently, of the process $\{ \epsilon_t \}_t$) produces two binary sequences $s$ and $s^\prime$ for which it holds $s_t = s_t^\prime$ for even $t$ and $s_t = - s_t^\prime$ for odd $t$. We therefore have a bijective correspondence between realisations of $\{ s_t \}$ and of $\{ s_t^\prime \}$ which also preserves the measure, that is, it holds
\begin{equation} \label{eq:isomorphism_prime}
	\mu_S (s_{t_1}^{t_k}) = \mu_{S^\prime} (s_{t_1}^{t_k \prime}) , \quad \text{for all } k \text{ and all } t_1 \leq \ldots \leq t_k .
\end{equation}
The following diagram provides a picture of the process isomorphism:
\begin{equation*}
	\begin{CD}
		\{ X_t \}  @> ^\prime >> \{ X_t^\prime \}\\
		 @V B VV                     @V B VV     \\
		\{ s_t \}  @> ^\prime >> \{ s_t^\prime \}\\
	\end{CD}
	.
\end{equation*}
From Equation (\ref{eq:isomorphism_prime}) it follows that $H_k^{AR(1)} (\phi) = H_k^{AR(1)} (- \phi)$, which means that (i) is proved.

Equality (ii) is proved in the very same way as for (i), by noting that the MA(1) process defined by $X_t^\prime = \epsilon_t^\prime - \theta \epsilon_{t-1}^\prime$, with $\epsilon_t^\prime = (-1)^t \epsilon_t$ for all $t$, is isomorphic to that defined by $X_t = \epsilon_t - \theta \epsilon_{t-1}$.

Finally, (iii) and (v) follow immediately from (i), while (iv) and (vi) follow immediately from (ii).
\end{proof}

\begin{proof}[Proof of Proposition \ref{prop:complement_string}]
If $\{ \epsilon_t \}_t$ is a Gaussian white noise process defining the process $\{ X_t \}_t$ (either AR(1) or MA(1)), the white noise $\bar{\epsilon} = \{ - \epsilon_t \}_t$ defines the process $\bar{X} = \{ - X_t \}_t$. This is actually isomorphic to the process $X$ itself, since the random variables $\epsilon_t$ have the same distributions as their opposites. The processes $S$ and $\bar{S}$, discretised versions of the processes $X$ and $\bar{X}$, are therefore isomorphic and the thesis follows.
\end{proof}

\begin{proof}[Proof of Proposition \ref{prop:mu(a_b)_AR(1)}]
The quantities $\mu (0 \cdot^i 0)$ and $\mu (0 \cdot^i 1)$ are the probabilities of the events $\{ X_1 < 0 \} \cap \{ X_{i+2} < 0 \}$ and $\{ X_1 < 0 \} \cap \{ X_{i+2} > 0 \}$, respectively. Recall that $X_{i+2} | X_1 \sim \calN \big (\phi^{i+1} X_1,\frac{1-\phi^{2(i+1)}}{1-\phi^2} \sigma^2 \big )$. Thus, proceeding as in Section \ref{sec:geometric_characterisation}, we are left with calculating the measures of the subsets of $\bbS^1$ cut by the lines in $\bbR^2$ given by equations $x_1 = 0$ and $\frac{\phi^{i+1}}{\sqrt{1-\phi^{2(i+1)}}} x_1 + x_2 = 0$. Equalities (\ref{eq:mu00_AR(1)}) and (\ref{eq:mu01_AR(1)}) follow immediately.
\end{proof}

\begin{proof}[Proof of Proposition \ref{prop:mu(a_b)_MA(1)}]
Just as in Proposition \ref{prop:mu(a_b)_AR(1)}, $\mu (00)$ and $\mu (01)$ are the probabilities of the events $\{ X_1 < 0 \} \cap \{ X_2 < 0 \}$ and $\{ X_1 < 0 \} \cap \{ X_2 > 0 \}$, respectively. Since the conditional distribution of $X_2 | X_1$ is $\calN (\frac{\theta}{1+\theta^2} X_1,(\frac{1+\theta^2+\theta^4}{1+\theta^2} \sigma^2))$, we have that $\mu (00)$ and $\mu (01)$ are the relative measures of the subsets of $\bbS^1$ cut by the lines of equations $x_1 = 0$ and $\frac{\theta}{\sqrt{1+\theta^2+\theta^4}} x_1 + x_2 = 0$. Expressions (\ref{eq:mu00_MA(1)}) and (\ref{eq:mu01_MA(1)}) follow straightforwardly.

Finally, equality (\ref{eq:mu_s1.s2_MA(1)}) is easily proved by noting that the conditional distribution of a random variable $X_t$, given $X_{t-i}$ with $i \geq 2$, is the same as its unconditional distribution because $X_t$ and $X_{t-i}$ ($i \geq 2$) are independent.
\end{proof}

%-----------------------------------------------------------
\section{Data cleaning} \label{sec:appendix_data_cleaning}

%---------------------------------------
\subsection{Outliers} \label{sec:outliers}
In order to detect and remove outliers, that may be present in high frequency data for example because of errors of transmission, we use the algorithm proposed in~\cite{Brownlees_Gallo:2006}. Though it was originally developed for tick-by-tick prices, we apply it to 1-minute data by setting the parameters suitably. The algorithm is designed to identify and remove the price records which are too distant from a mean value calculated in their neighbourhood. To be precise, a price $p_i$ in the price series is removed if
\begin{equation} \label{eq:outlier}
	|p_i - \bar{p}_i (k)| \geq c \, s_i (k) + \gamma ,
\end{equation}
where $\bar{p}_i (k)$ and $s_i (k)$ are respectively the $\delta$-trimmed\footnote{The $\delta / 2$ lowest and the $\delta / 2$ highest observations are discarded.} sample mean and sample standard deviation of the $k$ price records closest to time $i$, $c$ is a constant amplifying the standard deviation and $\gamma$ is a parameter that allows to avoid cases of zero variance (e.g., when there are $k$ equal consecutive prices).

We take $k = 20$, $\delta = 10\%$, $c = 5$, $\gamma = 0.05$. Results of this outlier detection procedure lead to the removal of a number of 1-minute observations ranging from 9 to 310 across the 55 ETFs. Their distribution with respect to the time of the day shows that these observations occur for the great majority at the very beginning and at the very end of the trading day, suggesting that the algorithm spuriously identifies as outliers some genuine observations where high variability is physiological. However, the number of 1-minute observations detected as outliers is so limited (about one every three days in the worst case) that even spurious removal has negligible impact on the results.

%---------------------------------------
\subsection{Stock splits} \label{sec:splits}
Price data made available by data providers are generally adjusted for stock splits\footnote{A (\emph{forward} or \emph{reverse}) \emph{stock split} is a change decided by the company both in the number of its shares and in the price of the single share, such that the market capitalisation remains constant. A stock split is said to be \emph{$m$-for-$n$} if $m$ new shares are emitted for every $n$ old ones, with a price adjustment from $p$ to $\frac{n}{m} p$. If $m > n$ it is called a \emph{forward stock split}, while if $m < n$ we have a \emph{reverse stock split} (or \emph{stock merge}).} To detect possible unadjusted splits, we check the condition
\begin{equation*}
	|r| > 0.2
\end{equation*}
in the return series. This procedure would detect, for example, a 3-for-2 split or a 4-for-5 merge.

In our data we find four unadjusted splits, which we pointwise remove from the return series.

%---------------------------------------
\subsection{Intraday volatility pattern} \label{sec:intraday_pattern}
As is well known, the volatility of intraday returns has a periodic behaviour. It is higher near the opening and the closing of the market, showing a typical U-shaped profile every day. Moreover, events like the release of news always at the same time, or the opening of another market, contribute to create a specific intraday pattern that is the basic volatility structure of each day. We filter out the intraday volatility pattern from the return series by using the following simple model with intraday volatility factors. If $R_{d,t}$ is the raw return of day $d$ and intraday time $t$, we define the rescaled return
\begin{equation} \label{eq:deseasonalised_return}
	\tilde{R}_{d,t} = \frac{R_{d,t}}{\zeta_t} ,
\end{equation}
where
\begin{equation} \label{eq:zeta_intraday_factors}
	\zeta_t = \frac{1}{N_{\text{days}}} \sum_{d'} \frac{|R_{d',t}|}{s_{d'}} ,
\end{equation}
where $N_{\text{days}}$ is the number of days in the sample and $s_{d'}$ is the standard deviation of absolute intraday returns of day $d'$.

In this paper, we refer to the rescaled returns $\tilde{R}$ defined by Equation (\ref{eq:deseasonalised_return}) as \emph{deseasonalised returns}.

As an example, we report in Figure \ref{fig:intraday_volatility_profile} the intraday volatility profile of the DIA 1-minute return series. 
\begin{figure}[ht!]
\centering
\includegraphics[width=\textwidth]{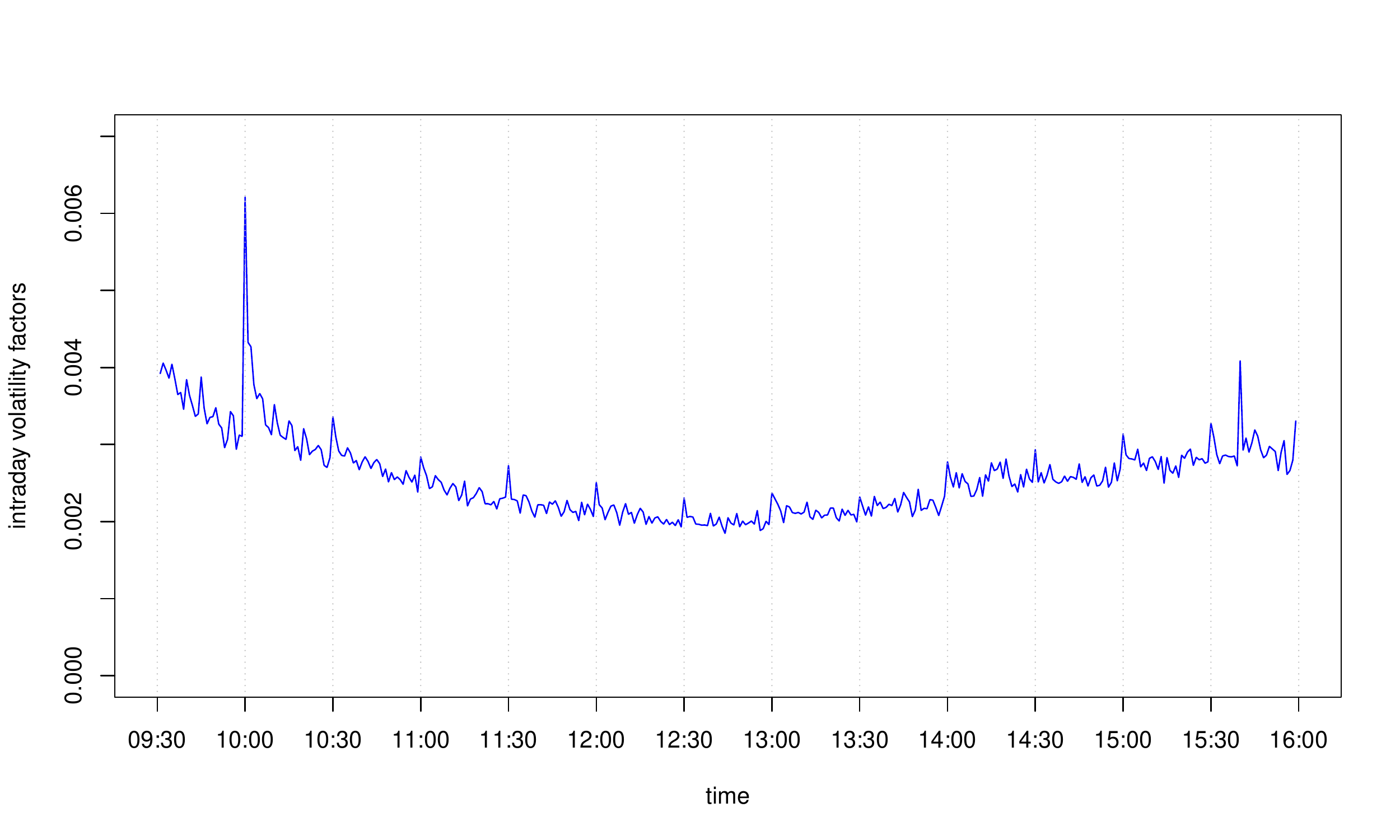}
\caption{Intraday volatility profile of 1-minute returns of the DIA ETF.}
\label{fig:intraday_volatility_profile}
\end{figure}

%---------------------------------------
\subsection{Heteroskedasticity} \label{sec:vol_proxy}
Deseasonalised return series, as defined by Equations (\ref{eq:deseasonalised_return}) and (\ref{eq:zeta_intraday_factors}), possess no residual intraday volatility structure, but they are still heteroskedastic, since different days can have different levels of volatility. In order to remove this heteroskedasticity, we estimate the time series of local volatility $\sigma_t$ and define the \emph{standardised returns} by
\begin{equation} \label{eq:return_standardisation}
	r_t = \frac{\tilde{R}_t}{\sigma_t} .
\end{equation}

As a proxy of the local volatility, we use the \emph{realised absolute variation} (see \cite{Andersen_Bollerslev:1997,Barndorff-Nielsen_Shephard:2004}). Let the logarithmic price $p (t)$ be generated by a process
\begin{equation}
	\dd p (t) = \mu (t) \, \dd t + \sigma (t) \, \dd W (t) ,
\end{equation}
where $\mu (t)$ is a finite variation process, $\sigma (t)$ a c\`adl\`ag volatility process and $W (t)$ a standard Brownian motion. Divide the interval $[0,t]$ into subintervals of the same length $\delta$ and denote by $r_i$ the return at time $i$, $p (i \delta) - p ((i-1) \delta)$. Then the following probability limit holds:
\begin{equation*}
	\text{p}-\lim_{\delta \searrow 0} \delta^{\frac{1}{2}} \sum_{i=1}^{\lfloor t/\delta \rfloor} |r_i| =
	\mu_1 \int_0^t \sigma (s) \, \dd s ,
\end{equation*}
where $\mu_1 = \bbE (|u|) = \sqrt{\frac{2}{\pi}} \simeq 0.797885$, $u \sim \calN (0,1)$.

Our estimator of local volatility is based on these quantity and is defined by the exponentially weighted moving average
\begin{equation} \label{eq:sigma_abs}
	\hat{\sigma}_{\text{abs},t} = \mu_1^{-1} \alpha \sum_{i > 0} (1-\alpha)^{i-1} |r_{t-i}| ,
\end{equation}
where $\alpha$ is the parameter of the exponential average to be specified. We take $\alpha = 0.05$ for the 1-minute data and $\alpha = 0.25$ for the 5-minute data, corresponding to a half-life time of nearly 14 minutes.

Filtering out the heteroskedasticity by means of Equation (\ref{eq:return_standardisation}), with the volatility estimated by Equation (\ref{eq:sigma_abs}), considerably reduces the excess kurtosis of the returns distribution for all the ETFs, thus proving to be an effective method. For instance, for the FXI ETF the excess kurtosis of 1-minute returns equals 11.87 before removing the heteroskedasticity and 0.88 after doing it. Figure \ref{fig:returns_histograms} shows the histograms of FXI intraday 1-minute returns, the intraday pattern being already removed, before and after the removal of heteroskedasticity by means of Equation (\ref{eq:return_standardisation}). As can be seen in the figures, there is a spike at 0 representing the great number of zero returns, due to the discreteness of price.
\begin{figure}[h]
\includegraphics[width=\textwidth]{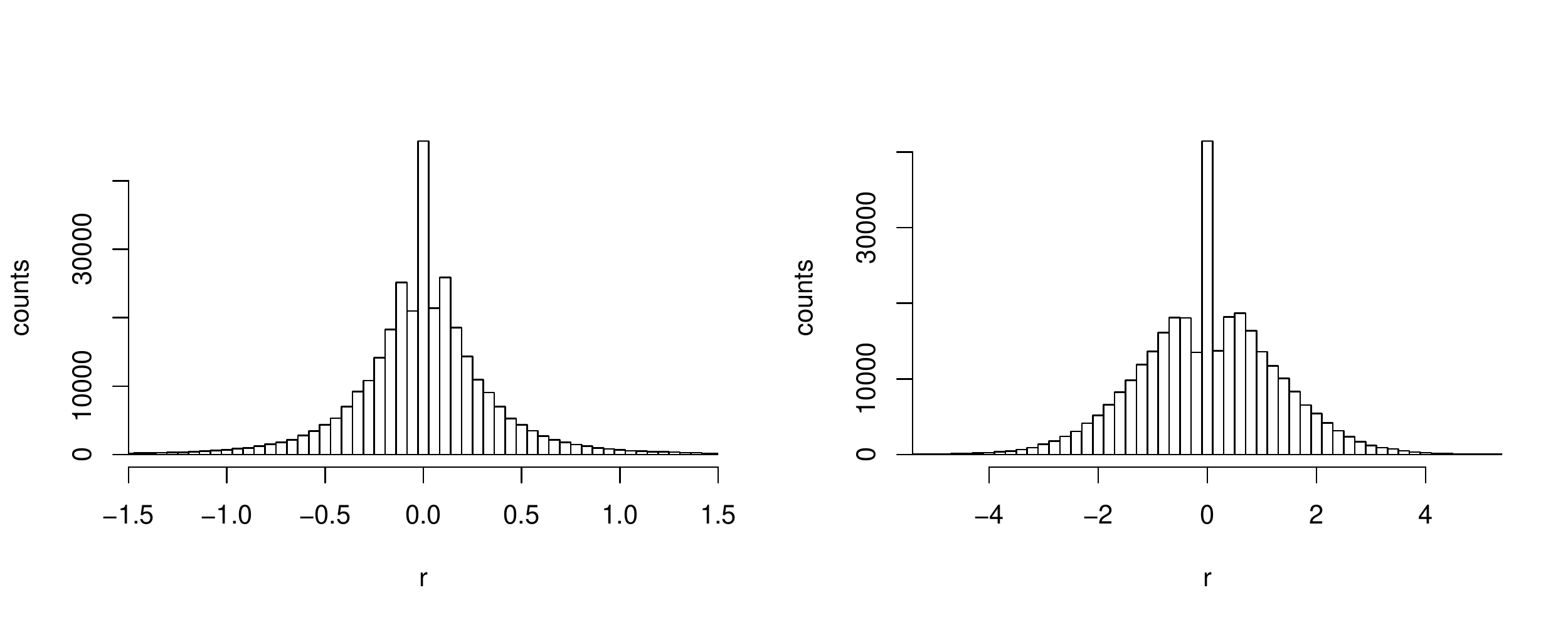}
\caption{Histograms of 1-minute returns of the FXI ETF, before (left) and after (right) removing the heteroskedasticity. The intraday pattern has already been filtered out in both series.}
\label{fig:returns_histograms}
\end{figure}

%---------------------------------------
\newpage
\bibliographystyle{mybibtexstyle}
\bibliography{biblio}
\addcontentsline{toc}{section}{Bibliography}

\end{document}